\documentclass[12pt,letterpaper]{article}
\usepackage{times}
\usepackage{paralist}
\usepackage{fullpage}
\usepackage{graphicx}
\usepackage{amsthm}
\usepackage{amsmath}
\usepackage{amsfonts}
\usepackage{amssymb}
\usepackage{dsfont}
\usepackage{color}
\usepackage{setspace}
\usepackage{url}
\usepackage{subfig}
\usepackage{caption}

\newcommand{\ignore}[1]{}

\DeclareMathOperator*{\conv}{\bigstar}

\newenvironment{proofof}[1]{\noindent \emph{ Proof of {#1}.}~}{\endproof}

\newenvironment{remark}[0]{\noindent \textbf{   Remark:}~}{\endproof}

\graphicspath{{figs/}}

\newtheorem{definition}{Definition}
\newtheorem{lemma}[definition]{Lemma}

\newtheorem{theorem}[definition]{Theorem}
\newtheorem{corollary}[definition]{Corollary}

\begin{document}

\renewcommand{\textfraction}{0}
\title{LP Decoding of Regular LDPC Codes in Memoryless Channels}

\author{
      Nissim Halabi \thanks{School of Electrical Engineering, Tel-Aviv University, Tel-Aviv 69978, Israel. \mbox{{E-mail}:\ {\tt nissimh@eng.tau.ac.il}.}}
      \and
      Guy Even \thanks{School of Electrical Engineering, Tel-Aviv University, Tel-Aviv 69978, Israel.  \mbox{{E-mail}:\ {\tt guy@eng.tau.ac.il}.}}
}

\date{October 28, 2010}

 \maketitle

\begin{abstract}
We study error bounds for linear programming decoding of regular LDPC
codes. For memoryless binary-input output-symmetric channels, we prove bounds on the word error probability that are inverse doubly-exponential in the girth of the factor graph. For memoryless binary-input AWGN
channel, we prove lower bounds on the threshold for regular LDPC codes whose factor graphs have logarithmic girth under LP-decoding.
Specifically, we prove a lower bound of $\sigma=0.735$ (upper bound of $\frac{Eb}{N_0}=2.67$dB) on the threshold of  $(3,6)$-regular LDPC codes whose factor graphs have logarithmic girth.

Our proof is an extension of a recent paper of Arora, Daskalakis, and
Steurer [STOC 2009] who presented a novel probabilistic analysis of LP
decoding over a binary symmetric channel. Their analysis is based on
the primal LP representation and has an explicit connection to message
passing algorithms. We extend this analysis to any MBIOS channel.
\end{abstract}

\section{Introduction} \label{sec:intro}
Low-density parity-check (LDPC) codes were invented by
Gallager~\cite{Gal63} in 1963. Gallager also invented the first type
of message-passing iterative decoding algorithm, known today as the
sum-product algorithm for a-posteriori probability (APP) decoding.
 Until the 1990s, iterative decoding systems were forgotten with a few
exceptions such as the landmark paper of Tanner~\cite{Tan81} in 1981,
who founded the study of codes defined by graphs. LDPC codes were
rediscovered~\cite{MN96} after the discovery of
turbo-codes~\cite{BGT93}. LDPC codes have attracted a lot of research attention
since empirical studies demonstrate excellent decoding performance
using iterative decoding methods. Among the main results is the
density-evolution technique for analyzing and designing asymptotic
LDPC codes~\cite{RU01}.  A density-evolution analysis computes a
threshold for the noise. This means that if the noise in the channel
is below that threshold, then the decoding error diminishes
exponentially as a function of the block length. The threshold results
of~\cite{RU01} hold for a random code from an ensemble of LDPC codes.

Feldman \emph{et al.}~\cite{Fel03,FWK05} suggested a decoding algorithm for linear codes that is based on linear programming. Initially, this idea seems to be counter-intuitive since codes are over $\mathds{F}_2^n$, whereas linear programming is over $\mathds{R}^n$. Following ideas from approximation algorithms, linear programming (LP) is regarded as a fractional relaxation of an integer program that models the problem of decoding. One can distinguish between integral solutions (vertices) and non-integral vertices of the LP. The integral vertices correspond to codewords, whereas the non-integral vertices are not codewords and are thus called pseudo-codewords. This algorithm, called LP-decoding, has two main advantages: (i) it runs in polynomial time, and (ii) when successful, LP-decoding provides an ML-certificate, i.e., a proof that its outcome agrees with maximum-likelihood (ML) decoding.

Koetter and Vontobel showed that LP-decoding is equivalent to graph cover decoding~\cite{VK05}. Abstractly, graph cover decoding proceeds
as follows.  Given a received word, graph cover decoding considers all
possible $M$-covers of the Tanner graph of the code (for every integer
$M$). For every $M$-cover graph, the variables are assigned $M$ copies
of the received word. Maximum-likelihood (ML) decoding is applied to
obtain a codeword in the code corresponding to the $M$-cover graph.
The ``best'' ML-decoding result is selected among all covers. This
lifted codeword is then projected (via averaging) to the base Tanner
graph. Obviously, this averaging might yield a non-integral solution,
namely, a pseudo-codeword as in the case of LP-decoding.
Graph cover decoding provides a combinatorial characterization of LP-decoding and pseudo-codewords.

LP-decoding has been applied to several codes, among them: RA codes, turbo-like codes, LDPC codes, and expander codes. Decoding failures have been characterized, and these characterizations enabled proving word error bounds for RA codes, LDPC codes, and expander codes (see e.g., \cite{FK04,HE05,KV06,FS05,FMSSW07,DDKW08,ADS09}).
Experiments indicate that message-passing decoding is likely to fail if LP-decoding fails~\cite{Fel03,VK05}.

\subsection{Previous Results}
Feldman \emph{et al.}~\cite{FMSSW07} were the first to show that LP-decoding corrects a constant fraction of errors for expander codes
over an adversarial bit flipping channel.  For example, for a specific
family of rate $\frac{1}{2}$ LDPC expander codes, they proved that LP-decoding can correct $0.000175n$ errors. This kind of analysis is
worst-case in its nature, and the implied results are quite far from
the performance of LDPC codes observed in practice over binary
symmetric channels (BSC). Daskalakis \emph{et al.}~\cite{DDKW08}
initiated an average-case analysis of LP-decoding for LDPC codes over
a probabilistic bit flipping channel. For a certain family of LDPC
expander codes over a BSC with bit flipping
probability $p$, they proved that LP-decoding recovers the transmitted
codeword with high probability up to a noise threshold of $p=0.002$.
This proved threshold for LP-decoding is rather weak compared to
thresholds proved for belief propagation (BP) decoding over the BSC.
For example, even for $(3,6)$-regular LDPC codes, the BP threshold is
$p=0.084$, and one would expect LDPC expander codes to be much
better under LP-decoding. Both of the results in \cite{FMSSW07} and \cite{DDKW08} were proved by analysis of the dual LP solution based on expansion arguments. Extensions of \cite{FMSSW07} to a larger class of channels (e.g., truncated AWGN channel) were discussed in \cite{FKV05}. 

Koetter and Vontobel~\cite{KV06} analyzed LP-decoding of regular LDPC codes using girth arguments and the dual LP solution.
They proved lower bound on the threshold of LP-decoding for regular LDPC codes whose Tanner graphs have logarithmic girth over any memoryless channel. This bound on the threshold depends only on the degree of the variable nodes. The decoding errors for noise below the threshold decrease doubly-exponentially in the girth of the factor graph.
This was the first threshold result presented for LP-decoding of LDPC codes over memoryless channels other than the BSC.
When applied to LP-decoding of
$(3,6)$-regular LDPC codes over a BSC with crossover probability $p$, they achieved a lower bound of $p=0.01$ on the threshold.
For the binary-input additive white Gaussian noise channel with noise variance $\sigma^2$
(BI-AWGN($\sigma$)), they achieved a lower bound of $\sigma = 0.5574$ on the threshold (equivalent to an upper bound of $\frac{Eb}{N_0}=5.07$dB).
The question of closing the gap to
$\sigma = 0.82$ ($1.7$dB)~\cite{WA01}, which is the threshold of max-product (min-sum) decoding algorithm for the same family of codes over a BI-AWGNC($\sigma$), remains open.

Recently, Arora \emph{et al.}~\cite{ADS09} presented a novel probabilistic analysis of the primal solution of LP-decoding for regular LDPC codes over a BSC using girth arguments. They proved error bounds that are inverse doubly-exponential in the girth of the Tanner graph and lower bounds on thresholds that are much closer to the performance of BP-based decoding.
For example, for a family of $(3,6)$-regular LDPC codes whose Tanner graphs have logarithmic girth over a BSC with crossover probability $p$, they proved a lower bound of $p=0.05$ on the threshold of LP-decoding. Their technique is based on a weighted decomposition of every  codeword and pseudo-codeword to a finite set of structured trees.
They proved a sufficient condition, called local-optimality, for the optimality of a decoded codeword based on this decomposition. They use a min-sum process on trees to bound the probability that local-optimality holds. A probabilistic analysis of the min-sum process is applied to the structured trees of the decomposition, and yields error bounds for LP-decoding.

In a following work, Vontobel~\cite{Von10} generalized the geometrical aspects presented by Arora \emph{et al.}~\cite{ADS09} to any code defined by a factor graph. Vontobel considered the general setup of factor graphs with (i) non-uniform node degrees, (ii) with other types of constraint function nodes, and (iii) with no restriction on the girth. Vontobel constructed a weighted decomposition of every codeword and pseudo-codeword to a finite set of structured combinatorial entities.

\subsection{Our Contribution}
In this work, we extend the analysis in~\cite{ADS09} from the BSC to any memoryless binary-input output-symmetric (MBIOS) channel. We prove bounds on the word error probability that are inverse doubly-exponential in the girth  of the factor graph for LP-decoding of regular LDPC codes over MBIOS channels.
We also prove lower bounds on the threshold of $(d_L,d_R)$-regular LDPC codes whose Tanner graphs have logarithmic girth under LP-decoding in binary-input AWGN channels. Note that regular Tanner graphs with logarithmic girth can be constructed explicitly (see e.g.~\cite{Gal63}).
Specifically, in a finite length analysis of LP-decoding over BI-AWGN($\sigma$), we prove that for $(3,6)$-regular LDPC codes the decoding errors for $\sigma<0.605$ ($\frac{Eb}{N_0}>4.36$dB) decrease doubly-exponentially in the girth of the factor graph. In an asymptotic case analysis, we prove a lower bound of $\sigma=0.735$ (upper bound of $\frac{Eb}{N_0}=2.67$dB) on the threshold of $(3,6)$-regular LDPC codes under LP-decoding, thus decreasing the gap to the BP-based decoding asymptotic threshold.

In our analysis we utilize the combinatorial interpretation of LP-decoding via graph covers~\cite{VK05} to simplify some of the proofs in~\cite{ADS09}. Specifically, using the equivalence of graph cover decoding and LP-decoding in \cite{VK05}, we obtain a simpler proof that local-optimality suffices for LP optimality.

Our main result:
\begin{theorem}\label{thm:main1}
Let $G$ denote a $(d_L,d_R)$-regular bipartite graph with girth $g$, and let $\mathcal{C}(G) \subset \{0,1\}^n$ denote the low-density parity-check code defined by $G$. Let $x \in \mathcal{C}(G)$ be a codeword. Consider the BI-AWGNC($\sigma$), and suppose that $y \in \mathds{R}^n$ is the word obtained from the channel given $x$. Then,
\begin{enumerate}
\item [1)][finite length bound] For $(d_L,d_R) = (3,6)$ and $\sigma \leqslant 0.605$ ($\frac{Eb}{N_0}\geqslant 4.36$dB), $x$ is the unique optimal solution to the LP decoder with probability at least \[1-\frac{1}{125}e^{\frac{3}{2\sigma^2}}n\cdot c^{2^{\lfloor\frac{1}{4}g\rfloor}}\] for some constant $c<1$.
\item [2)][asymptotic bound] For $(d_L,d_R) = (3,6)$ and $g=\Omega(\log n)$ sufficiently large, $x$ is the unique optimal solution to the LP decoder with probability at least $1-\exp(-n^\gamma)$ for some constant $0<\gamma<1$, provided that $\sigma\leqslant0.735$ ($\frac{Eb}{N_0}\geqslant2.67$dB).
\item [3)] For any $(d_L,d_R)$, $x$ is the unique optimal solution to the LP decoder with probability at least $1-n\cdot c^{(d_L-1)^{\lfloor\frac{1}{4}g\rfloor}}$ for some constant $c<1$, provided that
\begin{equation*}
  \min_{t\geqslant0}\bigg\{ \bigg( (d_R-1) e^{-t} \int_{-\infty}^{\infty}\big(1-F_{\mathcal{N}}(z)\big)^{d_R-2} f_{\mathcal{N}}(z)  e^{-tz}dz \bigg) \cdot \bigg((d_R-1) e^{\frac{1}{2}t^2\sigma^2-t} \bigg)^{1/(d_L-2)}\bigg\}<1,
\end{equation*}
where $f_{\mathcal{N}}(\cdot)$ and $F_{\mathcal{N}}(\cdot)$ denote the p.d.f. and c.d.f. of a Gaussian random variable with zero mean and standard deviation $\sigma$, respectively.
\end{enumerate}
\end{theorem}
Theorem~\ref{thm:main1} generalizes to MBIOS channels as follows.

\begin{theorem} \label{thm:main2} Let $G$ denote a $(d_L,d_R)$-regular
  bipartite graph with girth $\Omega(\log n)$, and let $\mathcal{C}(G)
  \subset \{0,1\}^n$ denote the low-density parity-check code defined
  by $G$. Consider an MBIOS channel, and suppose that $y \in
  \mathds{R}^n$ is the word obtained from the channel given $x=0^n$.  Let
  $\lambda \in \mathds{R}$ denote the log-likelihood ratio of the received
  channel observations, and let $f_\lambda(\cdot)$ and
  $F_\lambda(\cdot)$ denote the p.d.f. and c.d.f. of $\lambda(y_i)$,
  respectively.  Then,  LP-decoding succeeds with probability at least $1-\exp(-n^\gamma)$ for some
  constant $0<\gamma<1$, provided that
\begin{equation*}
  \min_{t\geqslant0}\bigg\{\bigg( (d_R-1) \int_{-\infty}^{\infty}\big(1-F_{\lambda}(z)\big)^{d_R-2} f_{\lambda}(z)  e^{-tz}dz \bigg) \cdot \bigg((d_R-1)\mathds{E}e^{-t\lambda} \bigg)^{1/(d_L-2)}\bigg\}<1.
\end{equation*}
\end{theorem}

The remainder of this paper is organized as follows. Section~\ref{sec:prelim} provides some background on low-density parity check codes and linear programming decoding over memoryless channels. Section~\ref{sec:optimality} presents combinatorial characterization of a sufficient condition of LP-decoding success for regular LDPC codes in memoryless channels. In section~\ref{sec:ErrorBound} we use the combinatorial characterization to bound the error probability of LP-decoding and provide lower bounds on the threshold. Thus proving Theorems~\ref{thm:main1} and~\ref{thm:main2}. We conclude with a discussion in Section~\ref{sec:discussion}.

\section{Preliminaries} \label{sec:prelim}

\paragraph{Low-density parity-check codes and factor graph representation.}
A code $\mathcal{C}$ with \emph{block length} $n$ over
$\mathds{F}_2$ is a subset of $\mathds{F}_2^n$. Vectors in
$\mathcal{C}$ are referred to as \emph{codewords}. An $[n,k]$
\emph{binary linear code} is a $k$-dimensional vector subspace of the vector space $\mathds{F}_2^n$. A \emph{parity-check matrix} for an $[n,k]$ binary linear code $\mathcal{C}$ is an $m \times n$ matrix $\mathbf{H}$ with $rank(\mathbf{H})=n-k\leqslant m$ whose rows span the space of vectors orthogonal to $\mathcal{C}$.

The \emph{factor graph} representation of a code $\mathcal{C}$ is a
bipartite graph $G$ that represents the matrix $\mathbf{H}$. The
factor graph $G$ is over \emph{variable nodes} $V_L \triangleq \{1,\ldots,n\}$
and \emph{check nodes} $V_R \triangleq \{1,\ldots,m\}$. An edge
$(i,j)$ connects variable node $i$ and check node $j$ if
$\mathbf{H}_{j,i}=1$. The variable nodes correspond to bits of the
codeword and the check nodes correspond to the rows of $\mathbf{H}$.
Every bipartite graph defines a parity check matrix. If the bipartite
graph is $(d_L,d_R)$-regular\footnote{That is, a bipartite graph with
  left vertices of degree $d_L$ and right vertices of degree $d_R$.}
for some constants $d_L$ and $d_R$, then it defines a
\emph{$(d_L,d_R)$-regular low-density parity-check (LDPC)} code.

\paragraph{LP decoding over memoryless channels.}
Let $X_i \in \{0,1\}$ and $Y_i \in \mathds{R}$ denote random variables that correspond to the $i$th transmitted symbol (channel input) and the $i$th received symbol (channel output), respectively. A \emph{memoryless binary-input output-symmetric} (MBIOS) channel is defined by a conditional probability density function
$f_{Y_i/X_i}(y_i/x_i) \triangleq f(Y_i=y_i/X_i=x_i)$ that satisfies $f_{Y_i/X_i}(y_i/0)=f_{Y_i/X_i}(-y_i/1)$. The \emph{log-likelihood ratio} (LLR)
vector $\lambda \in \mathds{R}^n$ for a received word $y \in \mathds{R}^n$ is defined by
\[ \lambda_i(y_i)
\triangleq \ln
\frac{f_{Y_i/X_i}(y_i/0)}{f_{Y_i/X_i}(y_i/1)},
\]
for $i \in \{1,\ldots,n\}$. For a linear code $\mathcal{C}$,
\emph{Maximum-Likelihood (ML) decoding} is equivalent to
\begin{equation} \label{eqn:MLdecoding}
 \hat{x}^{ML}(y) = \arg \min_{x \in \mathrm{conv}(\mathcal{C})} \langle
\lambda(y) , x \rangle,
\end{equation}
where $\mathrm{conv}(\mathcal{C})$ denotes the convex hull of the
set $\mathcal{C}$, where $\mathcal{C}$ is considered to be embedded in $\mathds{R}^n$ in the natural way.

Solving in general the optimization problem in (\ref{eqn:MLdecoding}) for linear codes is intractable. Furthermore, the decision problem of ML decoding remains NP-hard even for the class of left-regular LDPC codes~\cite{XH07}.
Feldman \emph{et al.}~\cite{Fel03,FWK05} introduced a linear programming relaxation for
the problem of ML decoding of linear codes. Given a factor graph
$G$, for every $j \in V_R$, denote by $\mathcal{C}_j$ the set of
binary sequences that satisfy parity check constraint $j$,
\[ \mathcal{C}_j \triangleq \big\{x \in \mathds{F}_2^n : \sum_{i \in \mathcal{N}(j)}x_i = 0 \mathrm{\ (mod 2)}  \big\}. \]
Let $\mathcal{P}(G) \triangleq \bigcap_{j \in
V_R}\mathrm{conv}(\mathcal{C}_j)$ denote the \emph{fundamental
polytope} \cite{Fel03,FWK05,VK05} of a factor graph $G$. For LDPC codes whose Tanner graphs have constant bounded right degree and a linear number of edges, the fundamental polytope can be defined by a linear number of constraints. Given an LLR vector $\lambda$ for a received word $y$, LP-decoding consists of solving the following
optimization problem
\begin{equation} \label{eqn:LPdecoding}
\hat{x}^{LP}(y) \triangleq \arg\min_{x \in \mathcal{P}(G)} \langle
\lambda(y) , x \rangle,
\end{equation}
which can be solved in time polynomial in $n$ using linear programming.

Let us denote by BI-AWGNC($\sigma$) the \emph{binary input additive white Gaussian noise channel} with noise variance $\sigma^2$. The channel input $X_i$ at time $i$ is an element of $\{\pm 1\}$ since we map a bit $b \in \{0,1\}$ to $(-1)^b$. Given $X_i$, the channel outputs $Y_i = X_i + \phi_i$ where $\phi_i \sim \mathcal{N}(0,\sigma^2)$. For BI-AWGNC($\sigma$),  $\lambda_i(y_i) = \frac{2y_i}{\sigma^2}$. Note that the optimal ML and LP solutions are invariant under positive scaling of the LLR vector $\lambda$.

\section{On the Connections between Local Optimality, Global Optimality, and LP Optimality} \label{sec:optimality}

Let $x \in \mathcal{C}(G)$ denote a codeword and $\lambda(y) \in \mathds{R}^n$ denote an LLR vector for a received word $y \in \mathds{R}^n$. Following \cite{ADS09}, we consider two questions: (i) does $x$ equal $\hat{x}^{ML}(y)$? and (ii) does $x$ equal $\hat{x}^{LP}(y)$ and is it the unique solution? Arora \emph{et al.}~\cite{ADS09} presented a certificate based on local
structures both for $\hat{x}^{ML}(y)$ and $\hat{x}^{LP}(y)$ over a binary symmetric channel. In this section we present modifications of definitions and certificates to the case of memoryless binary-input output-symmetric (MBIOS)
channels.

\emph{Notation:} Let $y \in \mathds{R}^n$ denote the received word.
Let $\lambda = \lambda(y)$ denote the LLR vector for $y$. Let $x \in
\mathcal{C}(G)$ be a candidate for $\hat{x}^{ML}(y)$ and
$\hat{x}^{LP}(y)$. $G$ is a $(d_L,d_R)$-regular bipartite factor
graph. For two vertices $u$ and $v$, denote by $d(u,v)$ the distance
between $u$ and $v$ in $G$. Denote by $\mathcal{N}(v)$ the set of
neighbors of a node $v$, and let $B(u,t)$ denote the set of vertices
at distance at most $t$ from $u$.

Following Arora \emph{et al.}\ we consider neighborhoods $B(i_0,2T)$ where $i_0 \in V_L$ and $T <
\frac{1}{4}girth(G)$. Note that the induced graph on $B(i_0,2T)$ is a tree.

\begin{definition} [Minimal Local Deviation, \cite{ADS09}]
\label{def:MinLocDev}
An assignment $\beta \in \{0,1\}^n$ is a
\emph{valid deviation of depth $T$ at $i_0 \in V_L$} or, in short, a
\emph{$T$-local deviation at $i_0$}, if $\beta_{i_0} =1$ and $\beta$
satisfies all parity checks in $B(i_0,2T)$,
\[ \forall j \in V_R \cap B(i_0,2T) : \sum_{i \in N(j)}\beta_i \equiv 0 \mod 2.  \]

A $T$-local deviation $\beta$ at $i_0$ is \emph{minimal} if
$\beta_i=0$ for every $i \notin B(i_0,2T)$, and every check node $j$
in $B(i_0,2T)$ has at most two neighbors with value $1$ in $\beta$.
A minimal $T$-local deviation at $i_0$ can be seen as a subtree of
$B(i_0,2T)$ of height $2T$ rooted at $i_0$, where every variable
node has full degree and every check node has degree 2. Such a tree
is called a \emph{skinny tree}. An assignment $\beta \in \{0,1\}^n$
is a minimal $T$-local deviation if it is a minimal $T$-local
deviation at some $i_0$. Note that given $\beta$ there is a unique
such $i_0 \triangleq \mathrm{root}(\beta)$.

If $w=(w_1,\ldots,w_T)\in [0,1]^T$ is a weight vector and $\beta$ is a
minimal $T$-local deviation, then $\beta^{(w)}$ denotes the
\emph{$w$-weighted} deviation
\[ \beta_i^{(w)} = \begin{cases} w_t\beta_i & \text{if $d(\mathrm{root}(\beta),i) = 2t$ and $1\leqslant t \leqslant T$,}\\
0 & \text{otherwise.} \end{cases} \]
\end{definition}

The following definition expands the notion of addition of codewords
over $\mathds{F}_2^n$ to the case where one of the vectors is
real.

\begin{definition} [\cite{Fel03}] \label{def:relPoint}
Given a codeword $x \in \{0,1\}^n$ and a point $f \in [0,1]^n$, the \emph{relative point} $x \oplus f \in [0,1]^n$ is defined by $(x \oplus f)_i = |x_i-f_i|$.
\end{definition}
Note that
\begin{equation*}
(x \oplus f)_i = \begin{cases} 1-f_i & \text{if $x_i=1$,}\\
f_i & \text{if $x_i = 0$.}
\end{cases}
\end{equation*}
Hence, for a fixed $x \in \{0,1\}^n$, $x \oplus f$ is an affine linear function in $f$. It follows that for any distribution over vectors $f \in [0,1]^n$, we have $\mathds{E}[x \oplus f] = x \oplus \mathds{E}[f]$.

Given a log-likelihood ratio vector $\lambda$, the cost of a $w$-weighted minimal $T$-local deviation $\beta$ is defined by $\langle \lambda, \beta^{(w)} \rangle$. The following definition is an extension of local-optimality from BSC to LLR.
\begin{definition} [local-optimality following \cite{ADS09}] \label{def{LocalOpt}}
A codeword $x \in \{0,1\}^n$ is \emph{$(T,w)$-locally optimal} for $\lambda \in \mathds{R}^n$ if for all minimal $T$-local deviations $\beta$,
\[ \langle \lambda , x \oplus \beta^{(w)} \rangle > \langle \lambda , x \rangle. \]
\end{definition}

Since $\beta^{(w)}\in [0,1]^n$, we consider only weight vectors $w
\in [0,1]^T\backslash\{0^n\}$. Koetter and Vontobel \cite{KV06}
proved for $w=1^T$ that a locally optimal codeword $x$ for
$\lambda$ is also globally optimal, i.e., the ML codeword. Moreover,
they also showed that a locally optimal codeword $x$ for $\lambda$
is also the unique optimal LP solution given $\lambda$. Arora
\emph{et al.}~\cite{ADS09} used a different technique to prove that
local-optimality is sufficient both for global optimality and LP
optimality with general weights in the case of a binary symmetric
channel. We extend the results of Arora \emph{et al.}~\cite{ADS09}
to the case of MBIOS channels. Specifically, we prove for MBIOS
channels that local-optimality implies LP optimality
(Theorem~\ref{thm:LPsufficient}). We first show how to extend the
proof that local-optimality implies global optimality in the case of MBIOS channels.

\begin{theorem} [local-optimality is sufficient for ML] \label{thm:MLsufficient}
Let $T < \frac{1}{4}girth(G)$ and $w \in [0,1]^T$.
Let $\lambda \in \mathds{R}^n$ denote the log-likelihood ratio for the received word, and suppose that $x \in \{0,1\}^n$ is a $(T,w)$-locally optimal codeword in $\mathcal{C}(G)$ for $\lambda$. Then $x$ is also the unique maximum-likelihood codeword for $\lambda$.
\end{theorem}

The proof for MBIOS channels is a straightforward modification of the proof in \cite{ADS09}. We include it for the sake of self-containment.
The following lemma is the key structural lemma in the proof of
Theorem~\ref{thm:MLsufficient}.
\begin{lemma} [\cite{ADS09}] \label{lemma:IntegralDecomposition}
Let $T < \frac{1}{4}girth(G)$. Then, for every codeword $z \neq 0^n$,
there exists a distribution over minimal $T$-local deviations
$\beta$ such that, for every weight vector $w \in [0,1]^T$, there
exists an $\alpha \in (0,1]$, such that  \[\mathds{E}_\beta
\beta^{(w)} = \alpha z.\]
\end{lemma}

\begin{proofof}{Theorem \ref{thm:MLsufficient}}
We want to show that for every codeword $x' \neq x$,  $\langle
\lambda , x' \rangle > \langle \lambda , x \rangle$. Since $z
\triangleq x \oplus x'$ is a codeword, by
Lemma~\ref{lemma:IntegralDecomposition} there exists a distribution
over minimal $T$-local deviations $\beta$ such that
$\mathds{E}_\beta \beta^{(w)} = \alpha z$. Let $f:[0,1]^n
\rightarrow \mathds{R}$ be the affine linear function defined by
$f(u) \triangleq \langle \lambda , x \oplus u \rangle = \langle
\lambda , x \rangle + \sum_{i=1}^{n}(-1)^{x_i}\lambda_i u_i$. Then,
\begin{eqnarray*}
  \langle \lambda , x \rangle &<& \mathds{E}_\beta \langle \lambda , x \oplus \beta^{(w)} \rangle \ \ \ \ (\text{by local-optimality of $x$}) \\
   &=& \langle \lambda , x \oplus \mathds{E}_\beta \beta^{(w)} \rangle \ \ \ \ (\text{by linearity of $f$ and linearity of expectation}) \\
   &=& \langle \lambda , x \oplus \alpha z \rangle  \ \ \ \ \ \ \ \ \ \ \ (\text{by Lemma \ref{lemma:IntegralDecomposition}})\\
   &=& \langle \lambda , (1-\alpha)x + \alpha(x \oplus z) \rangle\\
   &=& \langle \lambda , (1-\alpha)x + \alpha x' \rangle \\
   &=& (1-\alpha) \langle \lambda , x \rangle + \alpha \langle \lambda , x' \rangle.
\end{eqnarray*}
which implies that $\langle \lambda , x' \rangle > \langle \lambda , x \rangle$ as desired.
\end{proofof}

In order to prove a sufficient condition for LP optimality, we
consider graph cover decoding introduced by Vontobel and Koetter
\cite{VK05}. We use the terms and notation of Vontobel and Koetter
\cite{VK05} in the statement of Lemma~\ref{lemma:coverOptimality}
and the proof of Theorem \ref{thm:LPsufficient} (see Appendix
\ref{app:GCD}). The following lemma shows that local-optimality is
preserved after lifting to an $M$-cover. Note that the weight vector
must be scaled by the cover degree $M$.

\begin{lemma} \label{lemma:coverOptimality}
Let $T < \frac{1}{4}girth(G)$ and $w \in
[0,\frac{1}{M}]^T\backslash\{0^n\}$. Let $\tilde{G}$ denote any
$M$-cover of $G$. Suppose that $x \in \mathcal{C}(G)$ is a
$(T,w)$-locally optimal codeword for $\lambda \in \mathds{R}^n$. Let
$\tilde{x}=x^{\uparrow M}\in \mathcal{C}(\tilde{G})$ and
$\tilde{\lambda}=\lambda^{\uparrow M} \in \mathds{R}^{n\cdot M}$
denote the $M$-lifts of $x$ and $\lambda$, respectively. Then
$\tilde{x}$ is a $(T,M \cdot w)$-locally optimal codeword for
$\tilde{\lambda}$.
\end{lemma}

\begin{proof}
Assume that $\tilde{x}=x^{\uparrow M}$ is not a $(T,M \cdot
w)$-locally optimal codeword for $\tilde{\lambda}=\lambda^{\uparrow
M}$. Then, there exists a minimal $T$-local deviation $\tilde{\beta}
\in \{0,1\}^{n\cdot M}$ such that
\begin{equation} \label{eqn:proof1}
\langle \tilde{\lambda} , \tilde{x} \oplus \tilde{\beta}^{(M \cdot w)} \rangle \leqslant \langle \tilde{\lambda} , \tilde{x} \rangle.
\end{equation}
Note that for $\tilde{x}\in \{0,1\}^{n\cdot M}$ and its projection $x = p(\tilde{x})\in \mathds{R}^n$, it holds that
\begin{eqnarray}
   \frac{1}{M} \langle \tilde{\lambda}, \tilde{x} \rangle &=& \langle \lambda , x \rangle,  \mathrm{\ \ \ and} \label{eqn:proof2}\\
  \frac{1}{M} \langle \tilde{\lambda} , \tilde{x} \oplus \tilde{\beta}^{(M\cdot w)} \rangle &=& \langle \lambda , x \oplus \beta^{(w)} \rangle, \label{eqn:proof3}
\end{eqnarray}
where $\beta$ is the support of the projection of $\tilde{\beta}$
onto the base graph. It holds that $\beta$ is a $T$-local deviation because $T
< \frac{1}{4}girth(G) \leqslant
\frac{1}{4}girth(\tilde{G})$. From (\ref{eqn:proof1}),
(\ref{eqn:proof2}), and (\ref{eqn:proof3}) we get that $\langle \lambda
, x \rangle \geqslant \langle \lambda , x \oplus \beta^{(w)} \rangle$,
contradicting our assumption on the $(T,w)$-local optimality of $x$.
Therefore, $\tilde{x}$ is a $(T,M \cdot w)$-locally optimal codeword
for $\tilde{\lambda}$ in $\mathcal{C}(\tilde{G})$.
\end{proof}

Arora \emph{et al.}~\cite{ADS09} proved the following theorem for a
BSC and $w \in [0,1]^T$. The proof can be extended to the case of
MBIOS channels with $w \in [0,1]^T$ using the same technique of
Arora \emph{et al.}\ A simpler proof is achieved for $w \in
[0,\frac{1}{M}]^T$ for some finite $M$. The proof is based on arguments utilizing properties of graph cover decoding \cite{VK05}, and follows as a corollary of Theorem~\ref{thm:MLsufficient} and Lemma~\ref{lemma:coverOptimality}.

\begin{theorem} [local-optimality is sufficient for LP optimality] \label{thm:LPsufficient}
For every factor graph $G$, there exists a constant $M$ such that,
if
\begin{enumerate}
\item $T < \frac{1}{4}girth(G)$,
\item $w \in [0,\frac{1}{M}]^T\backslash \{0^T\}$, and
\item $x$ is a $(T,w)$-locally optimal codeword for
$\lambda \in \mathds{R}^n$,
\end{enumerate}
then $x$ is also the unique optimal LP solution given $\lambda$.
\end{theorem}

\begin{proof}
Suppose that $x$ is a $(T,w)$-locally optimal codeword for $\lambda
\in \mathds{R}^n$. Vontobel and Koetter \cite{VK05} proved that for
every basic feasible solution $z \in [0,1]^n$ of the LP, there
exists an $M$-cover $\tilde{G}$ of $G$ and an assignment $\tilde{z}
\in \{0,1\}^{n\cdot M}$ such that $\tilde{z} \in
\mathcal{C}(\tilde{G})$ and $z = p(\tilde{z})$, where $p(\tilde{z})$
is the image of the scaled projection of $\tilde{z}$ in $G$ (i.e., the
pseudo-codeword associated with $\tilde{z}$). Moreover, since the
number of basic feasible solutions is finite, we conclude that there
exists a finite $M$-cover $\tilde{G}$ such that every basic feasible
solution of the LP admits a valid assignment in $\tilde{G}$.

Let $z^*$ denote an optimal LP solution given $\lambda$. Without loss of generality $z^*$ is a basic feasible solution. Let $\tilde{z}^* \in \{0,1\}^{n\cdot M}$ denote the $0-1$ assignment in the $M$-cover $\tilde{G}$ that corresponds to $z^* \in [0,1]^n$. By the equivalence of LP-decoding and graph cover decoding~\cite{VK05}, (\ref{eqn:proof2}), and the optimality of $z^*$ it follows that $\tilde{z}^*$ is a codeword in $\mathcal{C}(\tilde{G})$ that minimizes $\langle \tilde{\lambda} , \tilde{z} \rangle$ for $\tilde{z} \in \mathcal{C}(\tilde{G})$, namely $\tilde{z}^* = \hat{x}^{ML}(y^{\uparrow M})$.

Let $\tilde{x}=x^{\uparrow M}$ denote the $M$-lift of $x$. Note that because $x$ is a codeword, i.e., $x \in \{0,1\}^n$, there is a unique pre-image of $x$ in $\tilde{G}$, which is the $M$-lift of $x$. Lemma~\ref{lemma:coverOptimality} implies that $\tilde{x}$ is a
$(T,M \cdot w)$-locally optimal codeword for $\tilde{\lambda}$,
where $M \cdot w \in [0,1]^T$. By Theorem~\ref{thm:MLsufficient}, we
also get that $\tilde{x} = \hat{x}^{ML}(y^{\uparrow M})$. Moreover,
Theorem~\ref{thm:MLsufficient} guarantees the uniqueness of an ML
optimal solution. Thus, $\tilde{x} = \tilde{z}^*$. By projection to
$G$, since $\tilde{x} = \tilde{z}^*$, we get that $x = z^*$ and
uniqueness follows, as required.
\end{proof}

From this point, let $M$ denote the constant whose existence is guaranteed by Theorem~\ref{thm:LPsufficient}.

\section{Proving Error Bounds Using Local Optimality} \label{sec:ErrorBound}
In order to simplify the probabilistic analysis of algorithms for decoding linear codes over symmetric channels, one can assume without loss of generality that the all-zero codeword was transmitted, i.e., $x = 0^n$. Note that the correctness of the all-zero assumption depends on the employed decoding algorithm. Although this assumption is trivial for ML decoding because of the symmetry of a linear code $\mathcal{C}(G)$, it is not immediately clear in the context of LP-decoding. Feldman \emph{et al.}~\cite{Fel03, FWK05} noticed that the fundamental polytope $\mathcal{P}(G)$ is highly symmetric, and proved that for binary-input output-symmetric channels, the probability that the LP decoder fails is independent of the transmitted codeword.
Therefore, one can assume that $x=0^n$ when analyzing LP-decoding failure for linear codes.
The following lemma gives a structural characterization for the event of LP-decoding failure if $x = 0^n$.

\begin{lemma} \label{lemma:LPfailure}
Let $T<\frac{1}{4} girth(G)$. Assume that the all-zero codeword was transmitted, and let $\lambda \in \mathds{R}^n$ denote the log-likelihood ratio for the received word. If the LP decoder fails to decode to the all-zero codeword, then for every $w \in \mathds{R}_+^T$ there exists a minimal $T$-local deviation $\beta$ such that $\langle \lambda, \beta^{(w)} \rangle \leqslant 0$.
\end{lemma}

\begin{proof}
Consider the event where the LP decoder fails to decode the all-zero codeword, i.e., $0^n$ is not a unique optimal LP solution. Theorem \ref{thm:LPsufficient} implies that there exists a constant $M$ such that, for every $w' \in [0,\frac{1}{M}]^T\backslash\{0^T\}$, the all-zero codeword is not the $(T,w')$-locally optimal codeword for $\lambda$. That is, there exists a minimal $T$-local deviation $\beta$ such that $\langle \lambda, \beta^{(w')} \rangle \leqslant 0$. Let $w' = \frac{1}{M\cdot ||w||_\infty}\cdot w$. Therefore $\langle \lambda, \beta^{(w)} \rangle$ is also non-positive, as required.
\end{proof}

We therefore have for a fixed $T<\frac{1}{4}girth(G)$ and $w \in \mathds{R}_+^T$ that
\begin{equation}\label{cor:LPfailureBound}
\mathbb{P}\{\mathrm{LP\ decoding\ fails}\} \leqslant \mathbb{P}\big\{\exists \beta \mathrm{\ such\ that\ } \langle \lambda, \beta^{(w)} \rangle \leqslant 0 \big| x=0^n\big\}.
\end{equation}

\subsection{Bounding Processes on Trees}
Using the terminology of~(\ref{cor:LPfailureBound}), Arora \emph{et al.}~\cite{ADS09} suggested a recursive method for bounding the probability $\mathbb{P}\big\{\exists \beta \mathrm{\ such\ that\ } \langle \lambda, \beta^{(w)} \rangle \leqslant 0 \big| x=0^n\big\}$ for a BSC.
We extend this method to MBIOS channels and apply it to a BI-AWGN channel.

Let $G$ be a $(d_L,d_R)$-regular bipartite factor graph, and fix $T<\frac{1}{4}girth(G)$. Let $\mathcal{T}_{v_0}$ denote the subgraph induced by $B(v_0,2T)$ for a variable node $v_0$. Since $T<\frac{1}{4}girth(G)$, it follows that $\mathcal{T}_{v_0}$ is a tree. We direct the edges of $\mathcal{T}_{v_0}$ so that it is an out-branching directed at the root $v_0$ (i.e., a rooted spanning tree with directed paths from the root $v_0$ to all the nodes).
For $l \in \{0,\ldots,2T\}$, denote by $V_l$ the set of vertices of $\mathcal{T}_{v_0}$ at height $l$ (the leaves have height $0$ and the root has height $2T$).
Let $\tau \subseteq V(\mathcal{T}_{v_0})$ denote the vertex set of a skinny tree rooted at $v_0$.

\begin{definition} [$(T,\omega)$-Process on a $(d_L,d_R)$-Tree, \cite{ADS09}] \label{def:process}
Let $\omega \in \mathds{R}_+^T$ denote a weight vector.
Let $\lambda$ denote an assignment of real values to the variable nodes of $\mathcal{T}_{v_0}$, we define the $\omega$-weighted value of a skinny tree $\tau$ by
\[ val_\omega (\tau;\lambda) \triangleq \sum_{l=0}^{T-1}\sum_{v \in \tau \cap V_{2l}}\omega_l \cdot \lambda_v. \]
Namely, the sum of the values of variable nodes in $\tau$ weighted according to their height.

Given a probability distribution over assignments $\lambda$, we are interested in the probability
\begin{equation}
\Pi_{\lambda,d_L,d_R}(T,\omega) \triangleq \mathbb{P}_\lambda \bigg\{ \min_{\tau \subset \mathcal{T}} val_\omega(\tau;\lambda) \leqslant 0 \bigg\} .
\end{equation}
\end{definition}
In other words, $\Pi_{\lambda,d_L,d_R}(T,\omega)$ is the probability that the minimum value over all skinny trees of height $2T$ rooted in some variable node $v_0$ in a $(d_L,d_R)$-bipartite graph $G$ is non-positive.
For every two roots $v_0$ and $v_1$ the trees $\mathcal{T}_{v_0}$ and $\mathcal{T}_{v_1}$ are isomorphic, it follows that $\Pi_{\lambda,d_L,d_R}(T,\omega)$ does not depend on the root $v_0$.

Since $\lambda$ is a random assignment of values to variable nodes in $\mathcal{T}_{v_0}$, Arora \emph{et al.}\ refer to $\min_{\tau \subset \mathcal{T}_{v_0}} val_\omega(\tau;\lambda)$ as a random process.
With this notation, we apply a union bound utilizing Lemma \ref{lemma:LPfailure}, as follows.
\begin{lemma}\label{lemma:LPsuccessLoose}
Let $G$ be a $(d_L,d_R)$-regular bipartite graph and $w \in \mathds{R}_+^T$ be a weight vector with $T<\frac{1}{4}girth(G)$. Suppose that $\lambda \in \mathds{R}^n$ is the log-likelihood ratio of the word received from the channel. Then, the transmitted codeword $x=0^n$ is $(T,\alpha \cdot w)$-locally optimal for $\alpha \triangleq (M\cdot||w||_\infty)^{-1}$ with probability at least
\[ 1-n \cdot \Pi_{\lambda,d_L,d_R}(T,\omega), \mathrm{\ \ \ \ where\ } \omega_l = w_{T-l},\]
and with at least the same probability, $x=0^n$ is also the unique optimal LP solution given $\lambda$.
\end{lemma}

Note the two different weight notations: (i) $w$ denotes weight vector in the context of weighted deviations, and (ii) $\omega$ denotes weight vector in the context of skinny subtrees in the $(T,\omega)$-Process. A one-to-one correspondence between these two vectors is given by $\omega_l = w_{T-l}$ for $0\leqslant l < T$. From this point on, we will use only $\omega$.

Following Lemma~\ref{lemma:LPsuccessLoose}, it is sufficient to estimate the probability $\Pi_{\lambda,d_L,d_R}(T,\omega)$ for a given weight vector $\omega$, a distribution of a random vector $\lambda$, and degrees $(d_L,d_R)$.
We overview the recursion presented in \cite{ADS09} for estimating and bounding the probability of the existence of a skinny tree with non-positive value in a $(T,\omega)$-process.

Let $\{\gamma\}$ denote an ensemble of i.i.d.\ random variables. Define random variables $X_0 ,\ldots, X_{T-1}$ and $Y_0 , \ldots , Y_{T-1}$ with the following recursion:
\begin{eqnarray}
  Y_0 &=& \omega_0 \gamma \label{eqn:Y_0}\\
  X_l &=& \min\big\{Y_l^{(1)},\ldots,Y_l^{(d_R-1)}\big\}\ \ \ \ \ \ \ \ \ \ \ (0 \leqslant l < T) \label{eqn:X_l}\\
  Y_l &=& \omega_l \gamma +  X_{l-1}^{(1)} +\ldots+X_{l-1}^{(d_L-1)}\ \ \ \ \ (0 < l < T) \label{eqn:Y_l}
\end{eqnarray}
The notation $X^{(1)},\ldots,X^{(d)}$ and $Y^{(1)},\ldots,Y^{(d)}$ denotes $d$ mutually independent copies of the random variables $X$ and $Y$, respectively. Each instance of $Y_l$, $0 \leqslant l < T$, uses an independent instance of a random variable $\gamma$.

Consider a directed tree $\mathcal{T}=\mathcal{T}_{v_0}$ of height $2T$, rooted at node $v_0$. Associate variable nodes of $\mathcal{T}$ at height $2l$ with copies of $Y_l$, and check nodes at height $2l+1$ with copies of $X_l$, for $0 \leqslant l <T$. Note that any realization of the random variables $\{\gamma\}$ to variable nodes in $\mathcal{T}$ can be viewed as an assignment $\lambda$. Thus, the minimum value of a skinny tree of $\mathcal{T}$ equals $\sum_{i=1}^{d_L}X_{T-1}^{(i)}$. This implies that the recursion in (\ref{eqn:Y_0})-(\ref{eqn:Y_l}) defines a dynamic programming algorithm for computing $\min_{\tau \subset \mathcal{T}} val_\omega(\tau;\lambda)$. Now, let the components of the LLR vector $\lambda$ be i.i.d.\ random variables distributed identically to $\{\gamma$\}, then
\begin{equation}
\Pi_{\lambda,d_L,d_R}(T,\omega) = \mathbb{P}\bigg\{\sum_{i=1}^{d_L}X_{T-1}^{(i)} \leqslant 0\bigg\}.
\end{equation}

Given a distribution of $\{\gamma\}$ and a finite ``height'' $T$, it is possible to compute the distribution of $X_l$ and $Y_l$ according to the recursion in (\ref{eqn:Y_0})-(\ref{eqn:Y_l}) using properties of a sum of random variables and a minimum of random variables (see Appendix~\ref{app:RVprop}). The following two lemmas play a major role in proving bounds on $\Pi_{\lambda,d_L,d_R}(T,\omega)$.
\begin{lemma}[\cite{ADS09}]\label{lemma:ADSbound1}
For every $t \geqslant 0$,
\[ \Pi_{\lambda,d_L,d_R}(T,\omega) \leqslant \big(\mathds{E}e^{-tX_{T-1}}\big)^{d_L}.\]
\end{lemma}
Let $d'_L \triangleq d_L-1$ and $d'_R \triangleq d_R-1$.
\begin{lemma}[\cite{ADS09}] \label{lemma:ADSbound2}
 For $0\leqslant s < l < T$, we have
 \[\mathds{E}e^{-tX_l} \leqslant {{\bigg(\mathds{E}e^{-tX_s} \bigg)}^{d'_L}}^{l-s}\cdot \prod_{k=0}^{l-s-1}{\big(d'_R\mathds{E}e^{-t\omega_{l-k}\gamma}\big)^{d'_L}}^k. \]
\end{lemma}

Based on these bounds, in the following subsection we present concrete bounds on $\Pi_{\lambda,d_L,d_R}(T,\omega)$ for BI-AWGN channel.

\subsection{Analysis for BI-AWGN Channel}

Consider the \emph{binary input additive white Gaussian noise channel} with noise variance $\sigma^2$ denoted by BI-AWGNC($\sigma$). In the case that the all-zero codeword is transmitted, the channel input is $X_i=+1$ for every $i$. Hence, $\lambda_i^{BI-AWGNC(\sigma)} = \frac{2}{\sigma^2}(1+\phi_i)$ where $\phi_i \sim \mathcal{N}(0,\sigma^2)$. Since $\Pi_{\lambda,d_L,d_R}(T,\omega)$ is invariant under positive scaling of the vector $\lambda$, we consider in the following analysis the scaled vector $\lambda$ in which $\lambda_i = 1+\phi_i$ with $\phi_i \sim \mathcal{N}(0,\sigma^2)$.

Following \cite{ADS09}, we apply a simple analysis for BI-AWGNC($\sigma$) with uniform weight vector $\omega$. Then, we present improved bounds by using a non-uniform weight vector.

\subsubsection{Uniform Weights}
Consider the case where $\omega = 1^T$. Let $c_1 \triangleq \mathds{E}e^{-tX_0}$ and $c_2 \triangleq d'_R\mathds{E}e^{-t\lambda_i}$, and define $c \triangleq c_1 \cdot c_2^{1/(d_L-2)}$.
By substituting notations of $c_1$ and $c_2$ in Lemmas~\ref{lemma:ADSbound1} and~\ref{lemma:ADSbound2},
Arora \emph{et al.}~\cite{ADS09} proved that if $c<1$, then
\[\Pi_{\lambda,d_L,d_R}(T,1^T) \leqslant c^{d_L \cdot {d'_L}^{T-1} - d_L }.\]

To analyze parameters for which $\Pi_{\lambda,d_L,d_R}(T,1^T)\rightarrow 0$, we need to compute $c_1$ and $c_2$ as functions of $\sigma$, $d_L$ and $d_R$.
Note that
\begin{eqnarray*}
X_0 &=& \min_{i\in \{1,\ldots,d'_R\}} \{ \lambda_i \} \\
&=&1 + \min_{i\in \{1,\ldots,d'_R\}} \phi_i,\ \ \mathrm{where\ } \phi_i \sim \mathcal{N}(0,\sigma^2) \ i.i.d.
\end{eqnarray*}
Denote by $f_{\mathcal{N}}(\cdot)$ and $F_{\mathcal{N}}(\cdot)$ the p.d.f. and c.d.f. of a Gaussian random variable with zero mean and standard deviation $\sigma$, respectively. We therefore have
\begin{eqnarray}
  c_1(\sigma, d_L, d_R) &=& d'_R e^{-t} \int_{-\infty}^{\infty}\big(1-F_{\mathcal{N}}(x)\big)^{d'_R-1} f_{\mathcal{N}}(x)  e^{-tx}dx ,\ \ \ \mathrm{and} \\
  c_2(\sigma, d_L, d_R)&=& d'_R e^{\frac{1}{2}t^2\sigma^2-t}.
\end{eqnarray}

The above calculations give the following bound on $\Pi_{\lambda,d_L,d_R}(T,1^T)$.

\begin{lemma} \label{lemma:uniformBound}
If $\sigma>0$ and $d_L,d_R>2$ satisfy the condition
\[ c = \min_{t\geqslant0}\bigg( \underbrace{d'_R e^{-t} \int_{-\infty}^{\infty}\big(1-F_{\mathcal{N}}(x)\big)^{d'_R-1} f_{\mathcal{N}}(x)  e^{-tx}dx}_{c_1} \bigg) \cdot \bigg(\underbrace{d'_R e^{\frac{1}{2}t^2\sigma^2-t}}_{c_2} \bigg)^{1/(d_L-2)}<1,
\]
then for $T \in \mathds{N}$ and $\omega = {1}^T$, we have
\[\Pi_{\lambda,d_L,d_R}(T,\omega) \leqslant c^{d_L \cdot {d'_L}^{T-1} - d_L }.\]
\end{lemma}

For (3,6)-regular graphs, we obtain by numeric calculations the following corollary.

\begin{corollary}
Let $\sigma<0.59$, $d_L = 3$, and $d_R=6$. Then, there exists a constant $c<1$ such that for every $T \in \mathds{N}$ and $\omega=1^T$,
\[\Pi_{\lambda,d_L,d_R}(T,\omega) \leqslant c^{2^T}.\]
\end{corollary}
\noindent Note that $\Pi_{\lambda,d_L,d_R}(T,1^T)$ decreases doubly-exponentially as a function of $T$.

\subsubsection{Improved Bounds Using Non-Uniform Weights}

The following lemma implies an improved bound for $\Pi_{\lambda,d_L,d_R}(T,\omega)$ using a non-uniform weight vector $\omega$.

\begin{lemma} \label{lemma:improvedBound}
Let $\sigma>0$ and $d_L,d_R>2$. Suppose that for some $s \in \mathds{N}$ and some weight vector $\overline{\omega}\in\mathds{R}_+^s$,
\begin{equation} \label{eqn:MinE}
 \min_{t \geqslant 0}\mathds{E}e^{-tX_s} < \big((d_R-1)e^{-\frac{1}{2\sigma^2}}\big)^{-\frac{1}{d_L-2}}.
\end{equation}
Let $\omega^{(\rho)} \in \mathds{R}^T_+$ denote the concatenation of the vector $\overline{\omega}\in\mathds{R}_+^s$ and the vector $(\rho,\ldots,\rho) \in \mathds{R}_+^{T-s}$. Then, for every $T>s$ there exist constants $c<1$ and $\rho \geqslant 0$ such that
\[\Pi_{\lambda,d_L,d_R}(T,\omega^{(\rho)}) \leqslant \big((d_R-1)e^{-\frac{1}{2\sigma^2}}\big)^{-\frac{d_L}{d_L-2}}\cdot c^{d_L \cdot {d'_L}^{T-s-1}}.\]
\end{lemma}

\begin{proof}
By Lemma~\ref{lemma:ADSbound2}, we have
\begin{eqnarray*}
\mathds{E}e^{-tX_{T-1}} &\leqslant& (\mathds{E}e^{-tX_s})^{(d_L-1)^{T-s-1}}\big((d_R-1)\mathds{E}e^{-t\rho(1+\phi)}\big)^{\sum_{k=0}^{T-s-2}(d_L-1)^k}\\
&=& (\mathds{E}e^{-tX_s})^{(d_L-1)^{T-s-1}}\big((d_R-1)\mathds{E}e^{-t\rho(1+\phi)}\big)^{\frac{(d_L-1)^{T-s-1}-1}{d_L-2}}.
\end{eqnarray*}
Note that $\mathds{E}e^{-t\rho(1+\phi)}=e^{-t\rho+\frac{1}{2}t^2\rho^2\sigma^2}$ is minimized when $t\rho = \sigma^{-2}$. By setting $\rho = \frac{1}{t\sigma^2}$, we obtain
\begin{eqnarray*}
\mathds{E}e^{-tX_{T-1}} &\leqslant& (\mathds{E}e^{-tX_s})^{(d_L-1)^{T-s-1}}\big((d_R-1)e^{-\frac{1}{2\sigma^2}} \big)^{\frac{(d_L-1)^{T-s-1}-1}{d_L-2}} \\
&=& \bigg(\mathds{E}e^{-tX_s} \big((d_R-1) e^{-\frac{1}{2\sigma^2}}\big)^{\frac{1}{d_L-2}}\bigg)^{(d_L-1)^{T-s-1}}\big((d_R-1)e^{-\frac{1}{2\sigma^2}}\big)^{-\frac{1}{d_L-2}}.
\end{eqnarray*}
Let $c \triangleq \big\{ \min_{t \geqslant 0}\mathds{E}e^{-tX_s} \big((d_R-1) e^{-\frac{1}{2\sigma^2}}\big)^{\frac{1}{d_L-2}}\big\}$. By (\ref{eqn:MinE}), $c<1$. Let $t^* = \arg \min_{t \geqslant 0}\mathds{E}e^{-tX_s}$, then
\begin{equation*}
 \mathds{E}e^{-t^*X_{T-1}} \leqslant c^{(d_L-1)^{T-s-1}}\big((d_R-1)e^{-\frac{1}{2\sigma^2}} \big)^{-\frac{1}{d_L-2}}.
\end{equation*}
Using Lemma~\ref{lemma:ADSbound1}, we conclude that
\begin{equation*}
 \Pi_{\lambda,d_L,d_R}(T,\omega^{(\rho)}) \leqslant c^{d_L(d_L-1)^{T-s-1}}\big((d_R-1)e^{-\frac{1}{2\sigma^2}} \big)^{-\frac{d_L}{d_L-2}},
\end{equation*}
and the lemma follows.
\end{proof}

Arora \emph{et al.}~\cite{ADS09} suggested using a weight vector $\overline{\omega}$ with components $\overline{\omega}_l = (d_L-1)^l$. This weight vector has the effect that if $\lambda$ assigns the same value to every variable node, then every level in a skinny tree $\tau$ contributes equally to $val_{\overline{\omega}} (\tau;\lambda)$.
For $T>s$, consider a weight vector $\omega^{(\rho)} \in \mathds{R}_+^T$ defined by
\[\omega_l = \begin{cases} \overline{\omega}_l & \mathrm{if\ }0 \leqslant l <s,\\
\rho & \mathrm{if\ } s\leqslant l < T .\end{cases}\] Note that the first $s$ components of $\omega^{(\rho)}$ are non-uniform while the other components are uniform.

For a given $\sigma$, $d_L$, and $d_R$, and for a concrete value $s$ we can compute the distribution of $X_s$ using the recursion in (\ref{eqn:Y_0})-(\ref{eqn:Y_l}). Moreover, we can also compute the value $\min_{t\geqslant0}\mathds{E}e^{-tX_s}$.
Computing the distribution and the Laplace transform of $X_s$ is not a trivial task in the case where the components of $\lambda$ have a continuous density distribution function.
However, since the Gaussian distribution function is smooth and most of its volume is concentrated in a defined interval, it is possible to ``simulate'' the evolution of the density distribution functions of the random variables $X_i$ and $Y_i$ for $i \leqslant s$. We use a numerical method based on quantization in order to represent and evaluate the functions $f_{X_l}(\cdot)$, $F_{X_l}(\cdot)$, $f_{Y_l}(\cdot)$, and $F_{Y_l}(\cdot)$. This computation follows methods used in the implementation of density evolution technique (see e.g. \cite{RU08}). A specific method for computation is described in Appendix~\ref{app:X_s} and exemplified for (3,6)-regular graphs.

For $(3,6)$-regular bipartite graphs we obtain the following corollary.
\begin{corollary} \label{cor:improvedBoundExample}
Let $\sigma<\sigma_0$, $d_L = 3$, and $d_R=6$.
For the following values of $\sigma_0$ and $s$ in Table~\ref{table:thresholds} it holds that there exists a constant $c<1$ such that for every $T>s$,
\[\Pi_{\lambda,d_L,d_R}(T,\omega) \leqslant \frac{1}{125}e^{\frac{3}{2\sigma^2}}\cdot c^{{2}^{T-s}}.\]

\begin{table}
\renewcommand{\arraystretch}{1.4}
\begin{center}
{\scriptsize
\begin{tabular}{||c|c|c||c||c|c|c||c||c|c|c||}
  \cline{1-3}
  \cline{5-7}
  \cline{9-11}
  $s$ & $\sigma_0$ & $\frac{E_b}{N_0}(\sigma_0)$[dB] & & $s$ & $\sigma_0$ & $\frac{E_b}{N_0}(\sigma_0)$[dB] & & $s$ & $\sigma_0$ & $\frac{E_b}{N_0}(\sigma_0)$[dB]\\
  \cline{1-3}
  \cline{5-7}
  \cline{9-11}
  \cline{1-3}
  \cline{5-7}
  \cline{9-11}
  $0$ & $0.605$ & $4.36$ & & $4$ & $0.685$ & $3.28$ & & $12$ & $0.72$ & $2.85$\\
  \cline{1-3}
  \cline{5-7}
  \cline{9-11}
  \cline{1-3}
  \cline{5-7}
  \cline{9-11}
  $1$ & $0.635$ & $3.94$& & $6$ & $0.7$ & $3.09$& & $14$ & $0.725$& $2.79$\\
  \cline{1-3}
  \cline{5-7}
  \cline{9-11}
  \cline{1-3}
  \cline{5-7}
  \cline{9-11}
  $2$ & $0.66$ & $3.60$& & $8$ & $0.71$& $2.97$ & & $18$ & $0.73$& $2.73$\\
  \cline{1-3}
  \cline{5-7}
  \cline{9-11}
  \cline{1-3}
  \cline{5-7}
  \cline{9-11}
  $3$ & $0.675$ & $3.41$& & $10$ & $0.715$ & $2.91$& & $22$ & $0.735$ & $2.67$\\
  \cline{1-3}
  \cline{5-7}
  \cline{9-11}
\end{tabular}
} 
\end{center}
\caption{Computed values of $\sigma_0$ for finite $s$ in Corollary~\ref{cor:improvedBoundExample}, and their corresponding $\frac{E_b}{N_0}$ SNR measure in dB.}\label{table:thresholds}
\end{table}
\end{corollary}

Note that for a fixed $s$, the probability $\Pi_{\lambda,d_L,d_R}(T,\omega)$ decreases doubly-exponentially as a function of $T$. Since it's required that $s<T$, Corollary~\ref{cor:improvedBoundExample} applies only to codes whose Tanner graphs have girth larger than $4T$.

Theorem~\ref{thm:main1} follows from Lemma~\ref{lemma:LPsuccessLoose}, Lemma~\ref{lemma:uniformBound}, and Corollary~\ref{cor:improvedBoundExample} as follows. The first part, that states a finite-length result, follows from Lemma~\ref{lemma:LPsuccessLoose} and Corollary~\ref{cor:improvedBoundExample} by taking $s = 0 < T < \frac{1}{4}girth(G)$ which holds for any Tanner graph $G$. The second part, that deals with an asymptotic result, follows from Lemma~\ref{lemma:LPsuccessLoose} and Corollary~\ref{cor:improvedBoundExample} by fixing $s=22$ and taking $g=\Omega(\log n)$ sufficiently large such that $s<T=\Theta(\log n) < \frac{1}{4}girth(G)$. It therefore provides a lower bound on the threshold of LP-decoding. The third part, that states a finite-length result for any $(d_L,d_R)$-regular LDPC code, follows from Lemma~\ref{lemma:LPsuccessLoose} and Lemma~\ref{lemma:uniformBound}. Theorem~\ref{thm:main2} is obtained in the same manner after a simple straightforward modification of Lemma~\ref{lemma:uniformBound} to MBIOS channels.

\begin{remark}
Following \cite{ADS09}, the contribution $\omega_T \cdot
\lambda_{v_0}$ of the root of $\mathcal{T}_{v_0}$ is not included in
the definition of $val_\omega (\tau;\lambda)$. The effect of this
contribution to $\Pi_{\lambda,d_L,d_R}(T,\omega)$ is bounded by a
multiplicative factor, as implied by the proof of
Lemma~\ref{lemma:ADSbound1}. The multiplicative factor is bounded by
$\mathds{E}e^{-t\omega_T\lambda_{v_0}}$, which may be regarded as a
constant since it does not depend on the code parameters (in
particular the code length $n$). Therefore, we can set $\omega_T = 0$ without loss of generality for these asymptotic considerations.
\end{remark}

\section{Discussion} \label{sec:discussion}

We extended the analysis of Arora \emph{et al.}~\cite{ADS09} for LP-decoding over a BSC to any MBIOS channel.
We proved bounds on the word error probability that are inverse doubly-exponential in the girth of the factor graph for LP-decoding of regular LDPC codes over MBIOS channels.
We also proved lower bounds on the threshold of regular LDPC codes whose Tanner graphs have logarithmic girth under LP-decoding in the binary-input AWGN channel.

Although thresholds are regarded as an asymptotic result, the
analysis presented by Arora \emph{et al.}~\cite{ADS09}, as well as its
extension presented in this paper, exhibits both asymptotic results as
well as finite-length results. An interesting tradeoff between these two
perspectives is shown by the formulation of the results. We regard the
goal of achieving the highest possible thresholds as an asymptotic
goal, and as such we may compare the achieved thresholds to the
asymptotic BP-based thresholds. Note that the obtained lower bound on the threshold increases
up to a certain ceiling value (which we conjecture is below the LP threshold) as the assumed girth increases.
Thus, an asymptotic result is obtained.

However, in the case of finite-length codes, the analysis cannot be based on an infinite girth in the limit. Two phenomena occur in the
analysis of finite codes: (i) the size of the interval $[0,\sigma_0]$ for which the error bound holds increases as function of
the girth (as shown in Table \ref{table:thresholds}), and (ii) the
decoding error probability decreases exponentially as a function of
the gap $\sigma_0-\sigma$ (as implied by Figure
\ref{fig:sigma_t_tradeoff}(b)).  We demonstrated the power of the analysis for the finite-length case by presenting error bounds for any $(3,6)$-regular LDPC code as function of the girth of the Tanner graph provided that $\sigma \leqslant 0.605$. Assuming that the girth of the Tanner graph is greater than 88, an error bound is presented provided that $\sigma \leqslant 0.735$. This proof also shows that $0.735$ is a lower bound on the threshold in the asymptotic case.

In the proof of LP optimality (Lemma~\ref{lemma:coverOptimality} and Theorem~\ref{thm:LPsufficient}) we used the combinatorial interpretation of LP-decoding via graph covers~\cite{VK05} to infer a reduction to conditions of ML optimality. That is, the decomposition of codewords presented by Arora \emph{et al.}~\cite{ADS09} leads to a decomposition for fractional LP solutions. This method of reducing combinatorial characterizations of LP-decoding to combinatorial characterizations of ML decoding is based on graph cover decoding.

\paragraph{Future directions:}
The technique for proving error bounds for BI-AWGN channel described in Section~\ref{sec:ErrorBound} and in Appendix~\ref{app:X_s} is based on a min-sum probabilistic process on a tree. The process is characterized by an evolution of probability density functions. Computing the evolving densities in the analysis of AWGN channels is not a trivial task. As indicated by our numeric computations, the evolving density functions in the case of the AWGN channel visually resemble Gaussian probability density functions (see Figures \ref{fig:X_evolve} and \ref{fig:Y_evolve}). Chung \emph{et al.}~\cite{CRU01} presented a method for estimating thresholds of belief propagation decoding according to density evolution using Gaussian approximation. Applying an appropriate Gaussian approximation technique to our analysis may result in analytic asymptotic approximate thresholds of LP-decoding for regular LDPC codes over AWGN channels.

Feldman \emph{et al.}~\cite{FKV05} observed that for high SNRs truncating LLRs of BI-AWGNC surprisingly assist LP-decoding. They proved that for certain families of regular LDPC codes and large enough SNRs (i.e., small $\sigma$), it is advantageous to truncate the LLRs before passing them to the LP decoder. The method presented in Appendix~\ref{app:X_s} for computing densities evolving on trees  using quantization and truncation of the LLRs can be applied to this case. It is interesting to see whether this unexpected phenomenon of LP-decoding occurs also for larger values of $\sigma$ (i.e., lower SNRs).

\appendix
\section{Graph Cover Decoding - Basic Terms and Notation} \label{app:GCD}

Vontobel and Koetter introduced in \cite{VK05} a combinatorial
concept called graph-cover decoding (GCD) for decoding codes on
graphs, and showed its equivalence to LP-decoding. The
characterization of GCD provides a useful theoretical tool for the
analysis of LP-decoding and its connections to iterative
message-passing decoding algorithms. We use the characterization of
graph cover decoding in the statement of
Lemma~\ref{lemma:coverOptimality} and the proof of Theorem
\ref{thm:LPsufficient}. In the following, we define some basic terms and
notations with respect to graph covers and graph-cover decoding.

Let $G$ and $\tilde{G}$ be finite graphs and let $\pi: \tilde{G}
\rightarrow G$ be a graph homomorphism, namely, $\forall \tilde{u},
\tilde{v} \in V(\tilde{G}):\ (\tilde{u},\tilde{v})\in E(\tilde{G})
\Rightarrow (\pi(\tilde{u}),\pi(\tilde{v})) \in E(G)$. A
homomorphism $\pi$ is a \emph{covering map} if for every $\tilde{v}
\in V(\tilde{G})$ the restriction of $\pi$ to neighbors of
$\tilde{v}$ is a bijection to the neighbors of $\pi(\tilde{v})$. The
pre-image $\pi^{-1}(v)$ of a node $v$ is called a \emph{fiber} and
is denoted by $\tilde{G}_v$. It is easy to see that all the fibers
have the same cardinality if $G$ is connected. This common
cardinality is called the \emph{degree} or \emph{fold number} of the
covering map. If $\pi: \tilde{G} \rightarrow G$ is a covering map,
we call $G$ the \emph{base graph} and $\tilde{G}$ a \emph{cover} of
$G$. In the case where the fold number of the covering map is $M$,
we say that $\tilde{G}$ is an $M$-cover of $G$.

Given a base graph $G$ and a natural fold number $M$, an $M$-cover
$\tilde{G}$ and a covering map $\pi: \tilde{G} \rightarrow G$ can be
constructed in the following way. Map every vertex $(v,i) \in
V(\tilde{G})$ (where $i \in \{1,\ldots,M\}$) to $v \in V(G)$, i.e.,
$\pi(v,i)=v$. The edges in $E(\tilde{G})$ are obtained by specifying
a matching $D_{(u,v)}$ of $M$ edges between $\pi^{-1}(u)$ and
$\pi^{-1}(v)$ for every $(u,v) \in E(G)$.

Note that the term `covering' originates from covering maps in
topology, as opposed to other notions of `coverings' in graphs or
codes (e.g., vertex covers or covering codes).

We now define assignments to variable nodes in an $M$-cover of a
Tanner graph. The assignment is induced by the covering map and an
assignment to the variable nodes in the base graph.

\begin{definition}[lift, \cite{VK05}] \label{def:lift}
Consider a bipartite graph $G = (\mathcal{I} \cup \mathcal{J}, E)$
and an arbitrary $M$-cover $\tilde{G} = (\tilde{\mathcal{I}} \cup
\tilde{\mathcal{J}}, \tilde{E})$ of  $G$. The \emph{$M$-lift} of a
vector $x \in \mathds{R}^N$ is an assignment $\tilde{x} \in
\mathds{R}^{N \cdot M}$ to the nodes in $\tilde{\mathcal{I}}$ that
is induced by the assignment $x \in \mathds{R}^N$ to the nodes in
$\mathcal{I}$ and the covering map $\pi : \tilde{G} \rightarrow G$ as follows:
 every $\tilde{v} \in \pi^{-1}(v)$ is assigned by $\tilde{x}$
the value assigned to $v$ by $x$. The $M$-lift of a vector $x$ is
denoted by $x^{\uparrow M}$.
\end{definition}

\begin{definition}[pseudo-codeword, \cite{VK05}] \label{def:pseudo-codeword}
The \emph{(scaled) pseudo-codeword} $p(\tilde{x}) \in \mathds{Q}^N$
associated with binary vector  $\tilde{x} = \{
\tilde{x}_{\tilde{v}}\}_{\tilde{v} \in \tilde{\mathcal{I}}} \in
\tilde{\mathcal{C}}$ of length $N \cdot M$ is the rational vector
$p(\tilde{x}) \triangleq
(p_{1}(\tilde{x}),p_{2}(\tilde{x}),\ldots,p_{N}(\tilde{x}))$ defined by
\begin{equation}
p_i(\tilde{x}) \triangleq \frac{1}{M}\cdot \sum_{\tilde{v} \in
\pi^{-1}(v_i)}\tilde{x}_{\tilde{v}},
\end{equation}
where the sum is taken in $\mathds{R}$ (not in $\mathds{F}_2$).
\end{definition}

\section{Computing the Evolution of Probability Densities over Trees}
\label{app:X_s}

In this appendix we present a computational method for estimating $\min_{t \geqslant 0}\mathds{E}e^{-tX_s}$ for some concrete $s$. The random variable $X_s$ is defined by the recursion in (\ref{eqn:Y_0})-(\ref{eqn:Y_l}).
Let $\{\gamma\}$ denote an ensemble of i.i.d.\ continuous random variable with probability density function (p.d.f.) $f_\gamma(\cdot)$ and cumulative distribution function (c.d.f.) $F_\gamma(\cdot)$.

We demonstrate the method for computing $\min_{t \geqslant 0}\mathds{E}e^{-tX_s}$ for the case where $d_L =3$, $d_R = 6$,
$\omega_l = (d_L-1)^l = 2^l$, $\sigma = 0.7$, and $\gamma = 1 +
\phi$ where $\phi \sim \mathcal{N}(0,\sigma^2)$. In this case,
\begin{eqnarray*}
f_\gamma(x) &=& f_{\mathcal{N}}(x-1) = \frac{1}{\sqrt{2\pi\sigma^2}}e^{-\frac{(x-1)^2}{2\sigma^2}},\ \ \ \mathrm{and}\\
F_\gamma(x) &=& F_{\mathcal{N}}(x-1) = \frac{1}{2}\bigg[1+\mathrm{erf} \bigg(\frac{x-1}{\sqrt{2}\sigma}\bigg)\bigg],
\end{eqnarray*}
where $\mathrm{erf}(x) \triangleq \frac{2}{\sqrt{\pi}}\int_{0}^{x}e^{{-t}^2}dt$ denotes the error function.

The actual computation of the evolution of density functions via the recursion equations requires a numeric implementation. Finding an efficient and stable such implementation is nontrivial. We follow methods used in the computation of the variable-node update process in the implementation of density evolution analysis (see e.g. \cite{RU08}).

We first state two properties of random variables for the evolving process defined in the recursion. We then show a method for computing a proper representation of the probability density function of $X_s$ for the purpose of finding $\min_{t \geqslant 0}\mathds{E}e^{-tX_s}$.

\subsection{Properties of Random Variables} \label{app:RVprop}
\paragraph{Sum of Random Variables.}
Let $\Phi$ denote a random variable that equals to the sum of $n$ independent random variables $\{\phi_i\}_{i=1}^n $, i.e.,
$\Phi = \sum_{i =1}^{n}\phi_i$. Denote by $f_{\phi_i}(\cdot)$ the p.d.f. of $\phi_i$. Then, the p.d.f. of $\Phi$ is given by
\begin{equation} \label{eqn:f_sumRV}
  f_\Phi = \conv_{i \in \{1,\ldots,n\}} {f_{\phi_i}},
\end{equation}
where $\star$ denotes the standard convolution operator over $\mathds{R}$ or over $\mathds{Z}$.

\paragraph{Minimum of Random Variables.}
Let $\Phi$ denote a random variable that equals to the minimum of $n$ i.i.d.\ random variables $\{\phi_i\}_{i=1}^n $, i.e.,
$\Phi = \min_{1 \leqslant i \leqslant n}\phi_i$. Denote by $f_\phi(\cdot)$ and $F_\phi(\cdot)$ the p.d.f. and c.d.f. of $\phi \sim \phi_i$, respectively. Then, the p.d.f. and c.d.f. of $\Phi$ are given by
\begin{eqnarray}
  f_\Phi(x) &=& n \cdot \big(1-F_\phi(x)\big)^{n-1}f_\phi(x),\ \ \ \mathrm{and} \label{eqn:f_minRV} \\
  F_\Phi(x) &=& 1 - \big(1-F_\phi(x)\big)^n. \label{eqn:F_minRV}
\end{eqnarray}

\subsection{Computing Distributions of $X_l$ and $Y_l$}

The base case of the recursion in (\ref{eqn:Y_0})-(\ref{eqn:Y_l}) is given by $Y_0$.
Let $g_{\omega_l}(\cdot)$ denote the p.d.f. of the scaled random variable $\omega_l\gamma$, i.e.,
\begin{equation} \label{eqn:g}
g_{\omega_l}(y) = \frac{1}{\omega_l}f_\gamma\bigg(\frac{y}{\omega_l}\bigg).
\end{equation}
Then, the p.d.f. of $Y_0$ is simply written as
\begin{equation} \label{eqn:f_Y0}
f_{Y_0}(y) = g_{\omega_0}(y).
\end{equation}
In the case where $\gamma = 1+ \mathcal{N}(0,\sigma^2)$, Equation~(\ref{eqn:f_Y0}) simplifies to
\begin{eqnarray}
f_{Y_0}(y) &=& \frac{1}{\omega_0}f_\mathcal{N}\bigg(\frac{y}{\omega_0}-1\bigg), \ \ \ \mathrm{and}\\
F_{Y_0}(y) &=& F_{\mathcal{N}}\bigg(\frac{y}{\omega_0}-1\bigg).
\end{eqnarray}

Let $f^{\bigstar d}(\cdot)$ denote the \emph{$d$-fold convolution} of a function $f(\cdot)$, i.e., the convolution of function $f(\cdot)$  with itself $d$ times.
Following (\ref{eqn:f_sumRV})-(\ref{eqn:F_minRV}), the recursion equations for the p.d.f. and c.d.f. of $X_l$ and $Y_l$ are given by
\begin{eqnarray}
f_{X_l}(x) &=& (d_R-1)\big(1-F_{Y_l}(x)\big)^{d_R-2}f_{Y_l}(x), \label{eqn:f_X_l}\\
F_{X_l}(x) &=& 1 - \big(1-F_{Y_l}(x)\big)^{d_R-1}, \label{eqn:F_X_l}\\
f_{Y_l}(y) &=& \bigg(g_{\omega_l} \star f_{X_{l-1}}^{\bigstar(d_L-1)}\bigg)(y),\ \ \
\mathrm{and} \label{eqn:f_Y_l}\\
F_{Y_l}(y) &=& \int_{-\infty}^{y}f_{Y_l}(t)dt. \label{eqn:F_Y_l}
\end{eqnarray}

Since we cannot analytically solve
(\ref{eqn:f_X_l})-(\ref{eqn:F_Y_l}), we use a numeric method
based on quantization in order to represent and evaluate the
functions $f_{X_l}(\cdot)$, $F_{X_l}(\cdot)$, $f_{Y_l}(\cdot)$, and
$F_{Y_l}(\cdot)$. As suggested in \cite{RU08}, we compute a uniform
sample of the functions, i.e., we consider the functions over the
set $\delta\mathds{Z}$, where $\delta$ denotes the quantization step
size. Moreover, due to practical reasons we restrict the functions
to a finite support, namely, $\{\delta k\}_{k=M}^N$ for some
integers $M<N$. We denote the set $\{\delta k\}_{k=M}^N$ by
$\delta[M,N]$. Obviously, the choice of $\delta$, $M$, and $N$
determines the precision of our target computation. Depending on the
quantized function, it is also common to consider point masses at
points not in $\delta[M,N]_{\mathds{Z}}$. For example, in the case where the density function has an heavy tail above $\delta N$ we may
assign the value $+\infty$ to the mass of the tail as an additional  quantization point. The same applies
analogously to a heavy tail below $\delta M$.

A Gaussian-like function (bell-shaped function) is bounded and
continuous, and so are its derivatives. The area beneath its tails
decays exponentially and becomes negligible a few standard
deviations away from the mean. Thus, Gaussian-like functions are
amenable to quantization and truncation of the tails. We therefore choose to zero the density functions outside the interval $[\delta M, \delta N]$. The parameters $M$ and $N$ are symmetric
around the mean, and together with $\delta$ are chosen to make the error of a Riemann integral negligible. As we demonstrate by computations,
the density functions $f_{X_l}(\cdot)$ and $f_{Y_l}(\cdot)$ are
indeed bell-shaped, justifying the quantization.
Figure~\ref{fig:f_X_0} illustrates the p.d.f. of $X_0$ (here $X_0$
equals to the minimum of $d_R-1 = 5$ instances of $Y_0$). Note that
by definition, $Y_0$ is a Gaussian random variable.

\begin{figure}
  \begin{center}
 \includegraphics[width=0.7\textwidth]{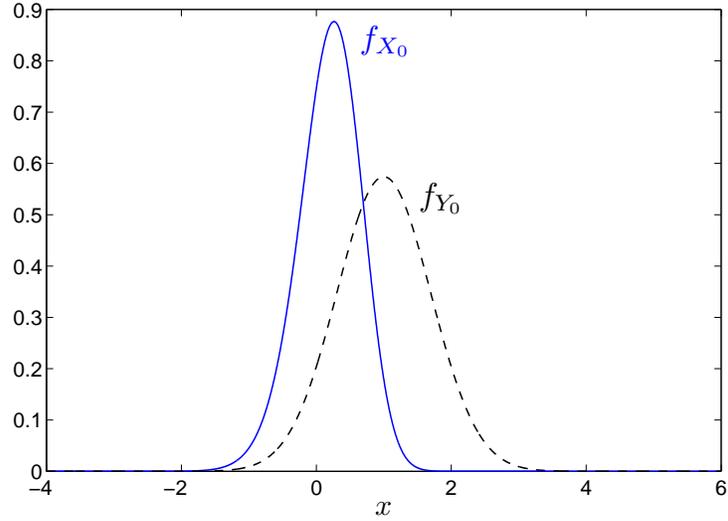}
  \caption{Probability density functions of $X_0$ and $Y_0$ for $(d_L,d_R)=(3,6)$ and $\sigma = 0.7.$}
  \label{fig:f_X_0}
  \end{center}
\end{figure}

Computing $f_{Y_l}(\cdot)$ given $f_{X_{l-1}}(\cdot)$ requires the convolution of functions. However, the restriction of the density
functions to a restricted support $\delta[M,N]$ is not invariant
under convolution. That is, if the function $f$ is supported by
$\delta[M,N]$, then $f\star f$ is supported by
$\delta[2M,2N]$. In the
quantized computations of $f_{X_l}(\cdot)$ and $f_{Y_l}(\cdot)$, our
numeric calculations show that the mean and standard deviation of
the random variables $X_l$ and $Y_l$ increase exponentially in $l$
as illustrated in Figures~\ref{fig:X_evolve} and~\ref{fig:Y_evolve}.
Therefore, the maximal slopes of the density functions $f_{X_l}(\cdot)$ and
$f_{Y_l}(\cdot)$ decrease with $l$. This property allows us to
double\footnote{Doubling applies to the demonstrated parameters,
i.e. $d_L=3$ and $\omega_l = 2^l$.} the quantization step $\delta$
as $l$ increases by one. Thus, the size of the support used for
$f_{X_l}(\cdot)$ and $f_{Y_l}(\cdot)$ does not grow. Specifically,
the interval $\delta[M,N]$ doubles but the doubling of $\delta$
keeps the number of points fixed. This method helps keep the
computation tractable while keeping the error small.

\begin{figure}
  \begin{center}
 \includegraphics[width=0.7\textwidth]{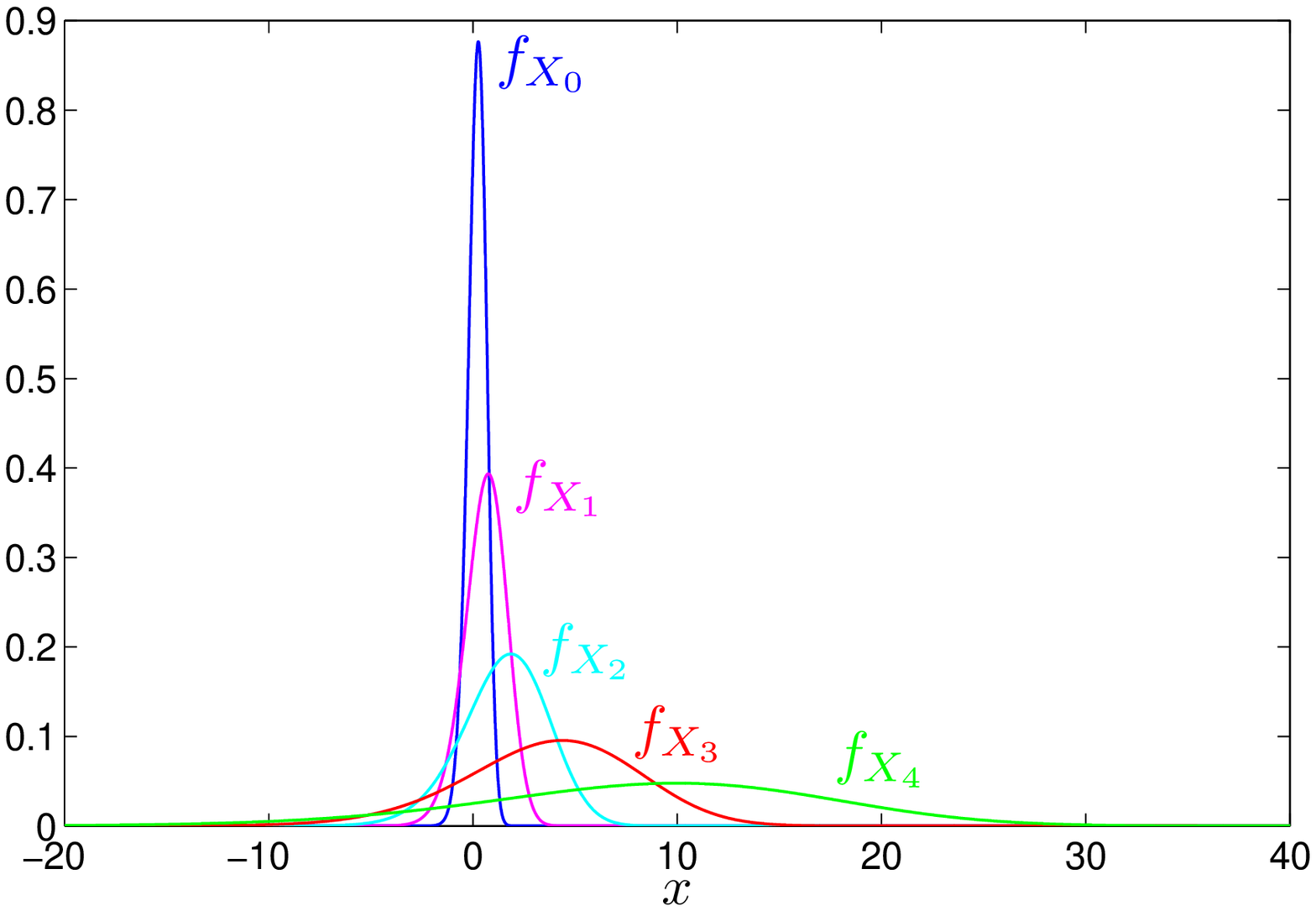}
  \caption{Probability density functions of $X_l$ for $l=0,\ldots,4$, $(d_L,d_R)=(3,6)$ and $\sigma = 0.7$.}
  \label{fig:X_evolve}
  \end{center}
\end{figure}

\begin{figure}
  \begin{center}
 \includegraphics[width=0.7\textwidth]{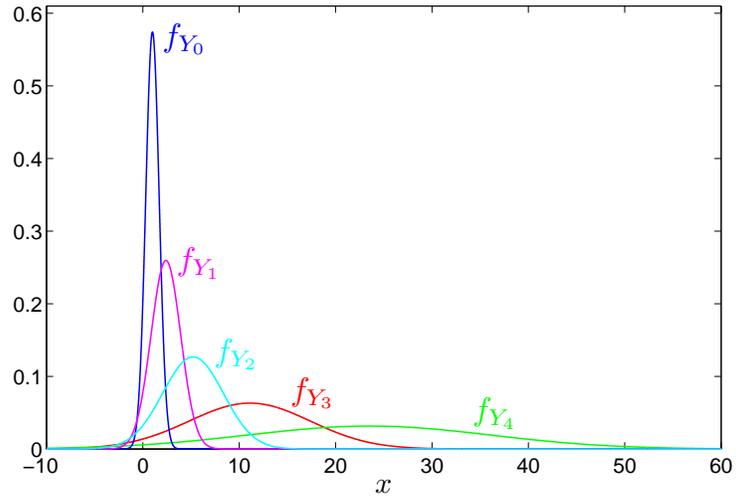}
  \caption{Probability density functions of $Y_l$ for $l=0,\ldots,4$, $(d_L,d_R)=(3,6)$ and $\sigma = 0.7$.}
  \label{fig:Y_evolve}
  \end{center}
\end{figure}

For two quantized functions $f$ and $g$, the calculation of $f\star
g$ can be efficiently performed using Fast Fourier Transform (FFT).
First, in order to prevent aliasing, extend the support with zeros
(i.e., zero padding) so that it equals the support of $f\star g$.
Then, $f\star g =\mathrm{IFFT}(\mathrm{FFT}(f)\times
\mathrm{FFT}(g))$ where $\times$ denotes a coordinate-wise
multiplication. The outcome is scaled by the quantization step size
$\delta$. In fact, the evaluation of $f_{Y_l}(\cdot)$ requires $d_L-1$
convolutions and is performed in the frequency domain (without
returning to the time domain in between) by a proper zero padding
prior to performing the FFT.

Note that when $\gamma$ is a discrete random variable with a bounded
support (as in \cite{ADS09}), a precise computation of the
probability distribution function of $X_s$ is obtained by following
(\ref{eqn:f_X_l})-(\ref{eqn:F_Y_l}).

\subsection{Estimating $\min_{t \geqslant 0}\mathds{E}e^{-tX_s}$}

After obtaining a proper discretized representation of the p.d.f. of
$X_s$ we approximate $\mathds{E}e^{-tX_s}$ for a given $t$ by
\begin{equation*}
\mathds{E}e^{-tX_s} \approxeq \sum_{k=M}^N \delta \cdot
f_{X_s}(\delta k) \cdot e^{-t\delta k}.
\end{equation*}
We then estimate the minimum value by searching over values of $t
\geqslant 0$. Figure~\ref{fig:E_e_X_s_evolve} depicts
$\ln\big(\mathds{E}e^{-tX_s}\big)$ as a function of $t \in (0,0.5]$
for $s=4,6,8,10,12$. The numeric calculations show that as $t$ grows
from zero, the function $\mathds{E}e^{-tX_s}$ decreases to a minimum
value, and then increases rapidly. We can also observe that both the
values $\min_{t \geqslant 0}\mathds{E}e^{-tX_s}$ and $\arg \min_{t
\geqslant 0}\mathds{E}e^{-tX_s}$ decrease as a function of $s$.

\begin{figure}
\centering \subfloat[][]{
\includegraphics[width=0.46\textwidth]{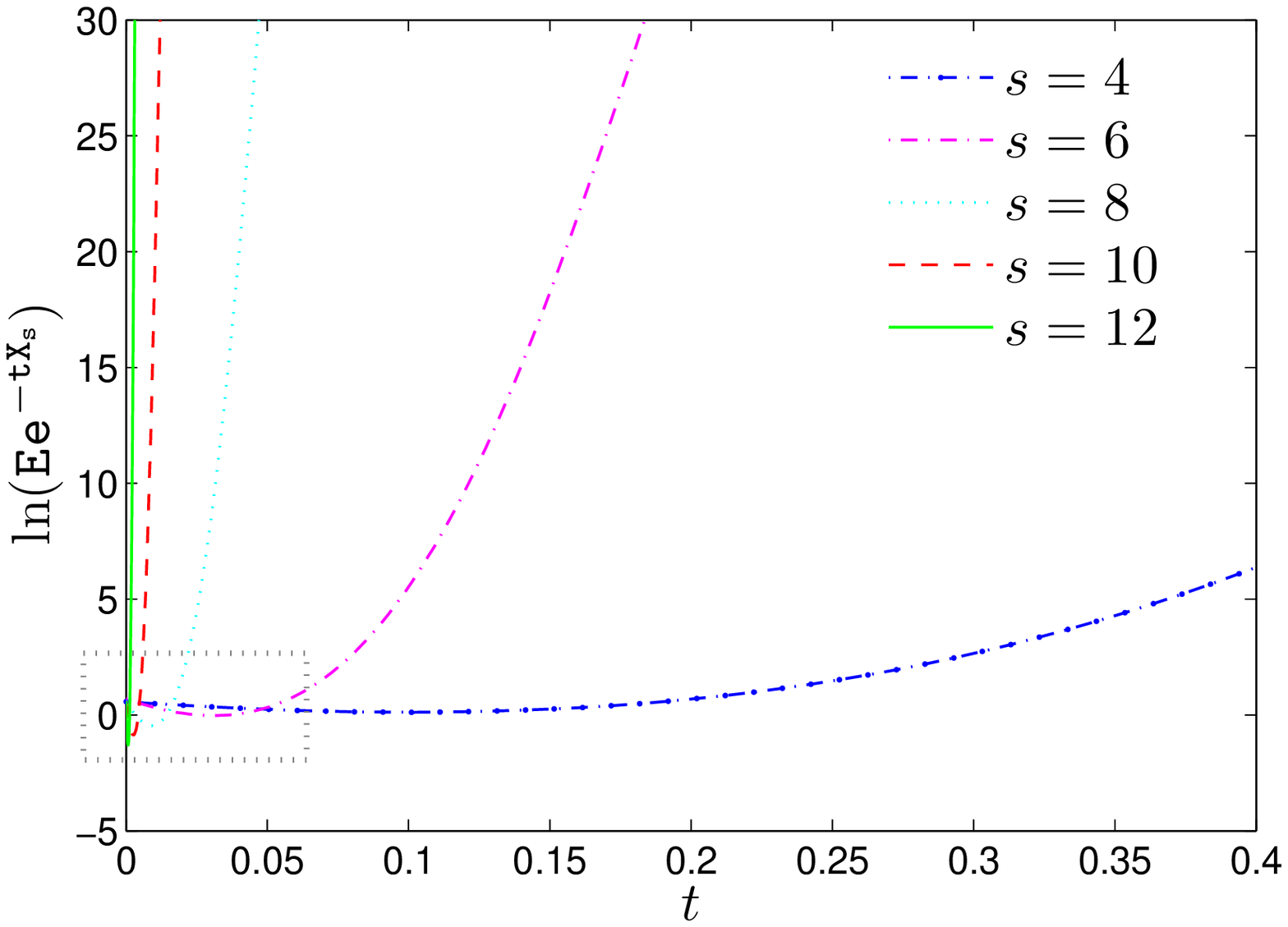}} \qquad \subfloat[][]{
\includegraphics[width=0.46\textwidth]{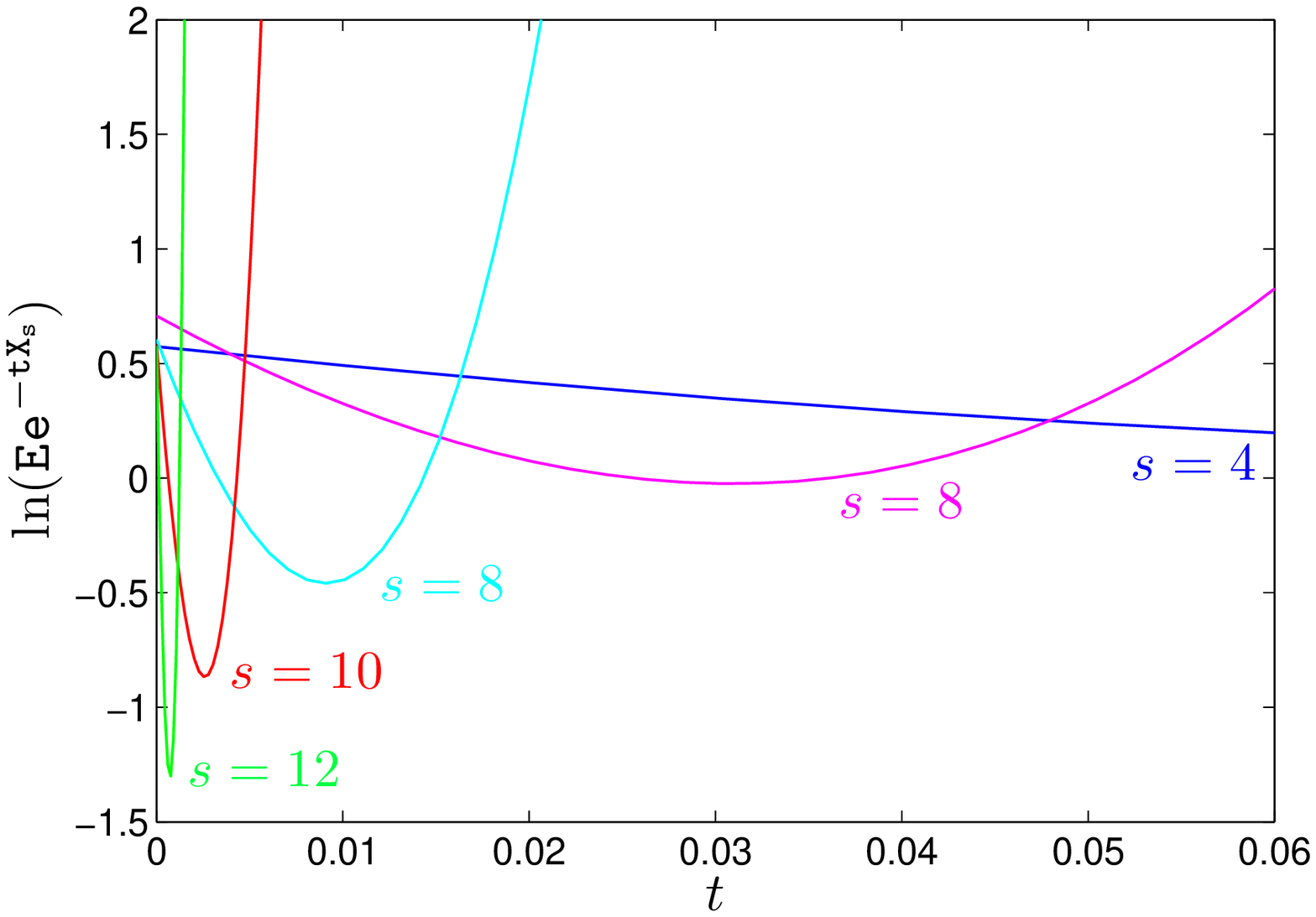}
} \caption{$\ln\big(\mathds{E}e^{-tX_s}\big)$ as a function of $t$
for $s=4,6,8,10,12$, $(d_L,d_R)=(3,6)$ and $\sigma = 0.7$. Plot (b)
is an enlargement of the rectangle depicted in plot (a).}
\label{fig:E_e_X_s_evolve}
\end{figure}

Following Lemma~\ref{lemma:improvedBound}, we are interested in the
maximum value of $\sigma$ for which (\ref{eqn:MinE}) holds
for a given $s$. That is,
\begin{equation} \label{eqn:sof}
\sigma_0 \triangleq \sup \bigg\{ \sigma>0 \ \ \bigg| \ \ \min_{t
\geqslant 0}\mathds{E}e^{-tX_s}\cdot
\big((d_R-1)e^{-\frac{1}{2\sigma^2}}\big)^{\frac{1}{d_L-2}} < 1
\bigg\}.
\end{equation}
Note that if the set in (\ref{eqn:sof}) is not empty, then it
is an open interval $(0,\sigma_0)\in \mathds{R}_+$.
Figure~\ref{fig:sigma_t_tradeoff}~(a) illustrates the region in the
$(t,\sigma)$ plane, for which (\ref{eqn:MinE}) holds with
$s=4$.

Let $t^*$ denote the value of $t$ that achieves the supremum
$\sigma_0$. For every $\sigma \in (0,\sigma_0)$, we may set the
value of the constant $c$ in
Corollary~\ref{cor:improvedBoundExample} as
\[c = \mathds{E}e^{-t^*X_s}\cdot \big((d_R-1)e^{-\frac{1}{2\sigma^2}}\big)^{\frac{1}{d_L-2}}.\]
Figure~\ref{fig:sigma_t_tradeoff}~(b) illustrates the value of the
constant $c$ in Corollary~\ref{cor:improvedBoundExample} as a
function of $\sigma$ in the case where $s=4$.

\begin{figure}
  \centering \subfloat[][]{
 \includegraphics[width=0.38\textwidth]{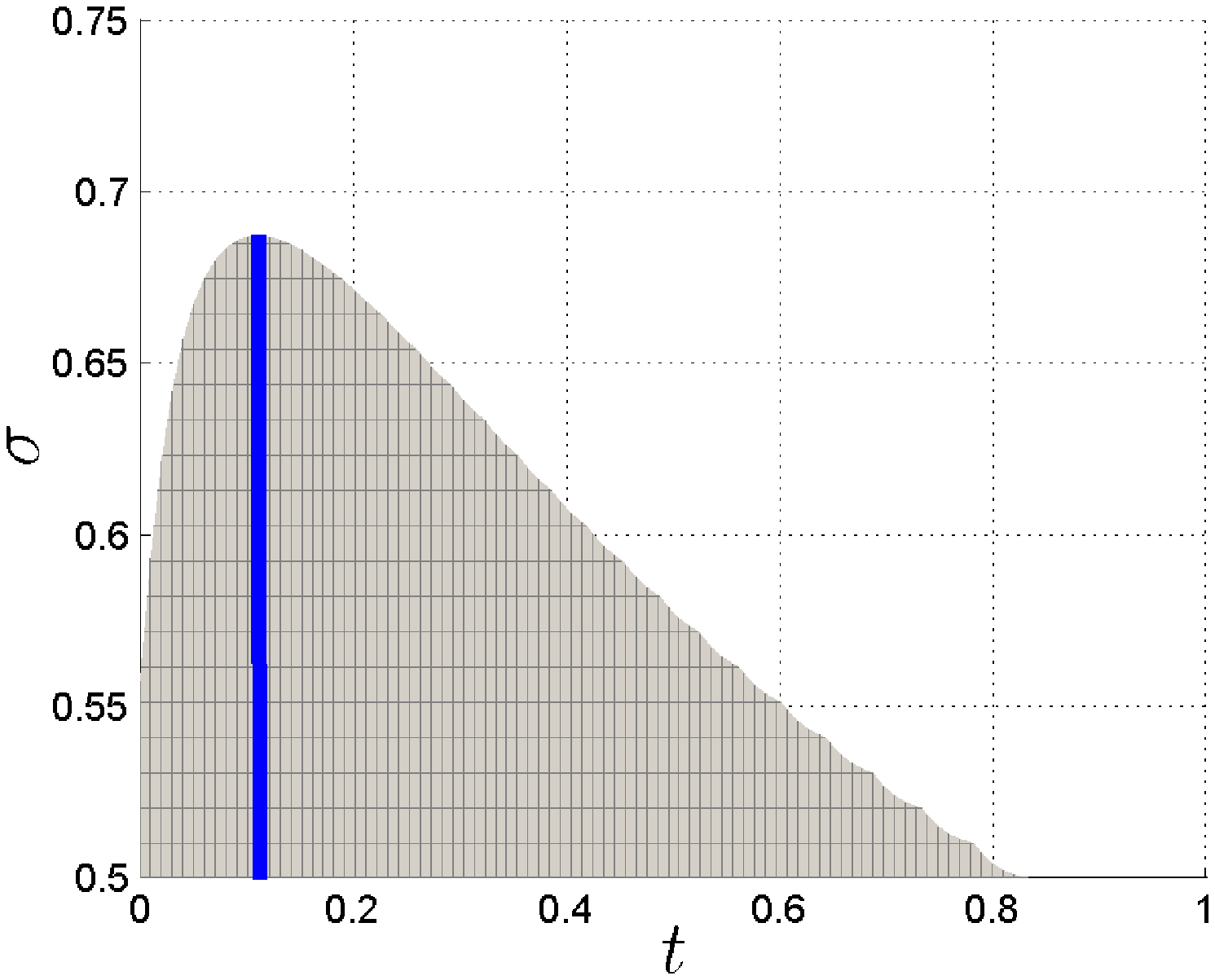}} \qquad \subfloat[][]{
\includegraphics[width=0.48\textwidth]{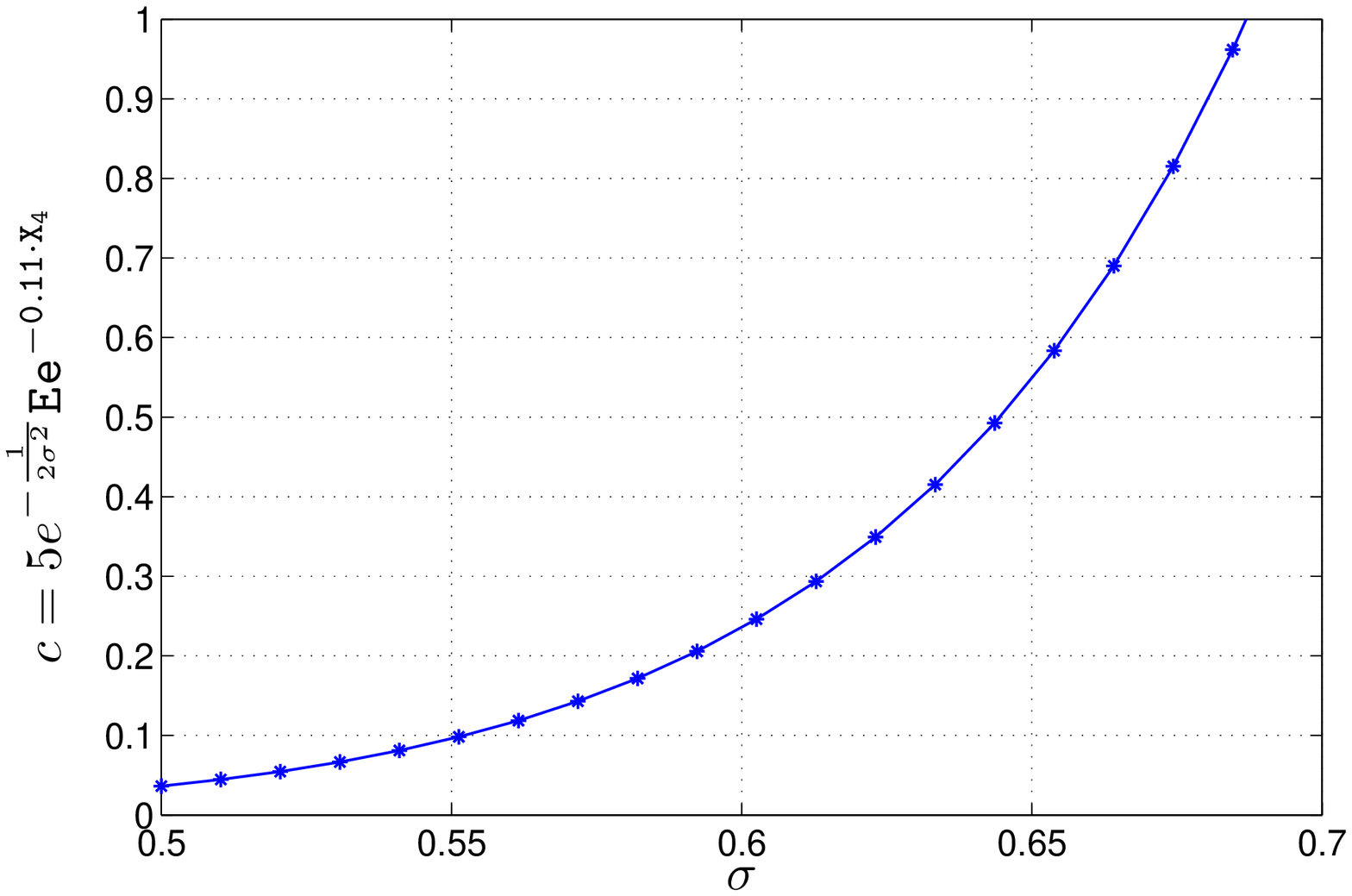}}
  \caption{(a) Region for which $5e^{-\frac{1}{2\sigma^2}}\mathds{E}e^{-tX_4} < 1$
  as a function of $t$ and $\sigma$ for $(d_L,d_R)=(3,6)$. Note that the maximal value of $\sigma$ contained in that region results to the estimate of $\sigma_0=0.685$ in the entry $s=4$ in Table~\ref{table:thresholds}.
  (b) Constant $c$ in Corollary~\ref{cor:improvedBoundExample} as a
function of $\sigma$ in the case where $s=4$ and $t=0.11$, i.e., the
value of $c$ over the cut of the $(t,\sigma)$-plane in plot (a) at
$t=0.11$ (depicted by a thick solid line).}
\label{fig:sigma_t_tradeoff}
\end{figure}

\section*{Acknowledgment}

The authors would like to thank P.~O.~Vontobel and the anonymous reviewers for their constructive comments and suggestions that helped to improve the paper.



\newcommand{\etalchar}[1]{$^{#1}$}

\end{document}